\documentclass{article}
\usepackage[utf8]{inputenc}

\usepackage{wrapfig}
\usepackage{graphicx}
\usepackage{xcolor}

\usepackage[letterpaper,margin=1in]{geometry}

\usepackage{amsmath}
\usepackage{natbib}
\usepackage{amsfonts}
\usepackage{amsthm}
\usepackage{graphicx}
\usepackage{subcaption}
\usepackage{multirow}
\usepackage{wrapfig}
\usepackage{url}
\usepackage{xcolor}

\newcommand{\xhdr}[1]{\vspace{1mm} \noindent\textbf{#1.}}
\newcommand{\var}{\mathrm{var}}

\newtheorem{remark}{Remark}

\newtheorem{theorem}{Theorem}
\newtheorem{proposition}{Proposition}

\newtheorem{lemma}{Lemma}
\newtheorem{corollary}{Corollary}

\newcommand{\Exp}{\mathbb{E}}

\newcommand{\reals}{\mathbb{R}}

\newcommand{\vx}{\mathbf{x}}

\newcommand{\hy}{\hat{y}}

\newcommand{\vtheta}{\boldsymbol{\theta}}

\newcommand{\zeros}{\textbf{0}}

\newcommand{\cG}{\mathcal{G}}

\newcommand{\cN}{\mathcal{N}}

\newcommand{\cY}{\mathcal{Y}}
\newcommand{\cX}{\mathcal{X}}

\newcommand{\cL}{\mathcal{L}}

\usepackage[suppress]{color-edits}
\addauthor[Karen]{kl}{purple}
\addauthor[Hoda]{hh}{red}
\addauthor[Solon]{sb}{blue}
\addauthor[Jon]{jk}{magenta}

\author{}

\author{
   \makebox[.45\textwidth]{Hoda Heidari}\\
   Carnegie Mellon University\\
   \url{hheidari@cmu.edu} \\
   \and
   \makebox[.45\textwidth]{Solon Barocas}\\
   Cornell University\\
   \url{sbarocas@cornell.edu} \\
    \and
   \makebox[.45\textwidth]{Jon Kleinberg}\\
   Cornell University\\
   \url{kleinberg@cornell.edu} \\
      \and
   \makebox[.45\textwidth]{Karen Levy}\\
   Cornell University\\
   \url{karen.levy@cornell.edu} \\
}

\title{Informational Diversity and Affinity Bias\\ in Team Growth Dynamics}
\date{}

\begin{document}

\maketitle

\begin{abstract}
Prior work has provided strong evidence that, within organizational settings, teams that bring a diversity of information and perspectives to a task are more effective than teams that do not.
If this form of \emph{informational diversity} confers performance advantages, why do we often see largely homogeneous teams in practice?  One canonical argument is that the benefits of informational diversity are in tension with \emph{affinity bias}.
To better understand \hhedit{the impact of this tension on the makeup of teams,}\hhdelete{this tension,} we analyze a sequential model of team formation in which individuals care about their team's performance (captured in terms of accurately predicting some future outcome based on a set of features) but experience a cost as a result of interacting with teammates who use different approaches to the prediction task. 
\hhedit{Our analysis of this simple model reveals a set of subtle behaviors that team-growth dynamics can exhibit: (i) from certain initial team compositions, they can make progress toward better performance but then get stuck partway to optimally diverse teams; while (ii) from other initial compositions, they can also move away from this optimal balance as the majority group tries to crowd out the opinions of the minority. The initial composition of the team can determine whether the dynamics will move toward or away from performance optimality, painting a \emph{path-dependent} picture of \emph{inefficiencies} in team compositions.}
\hhdelete{
Through this lens, we \hhedit{observe} that 
more diverse teams perform better but are less happy with their within-team interactions. \hhedit{Subsequently}, our analysis identifies a \hhedit{\emph{path-dependent}} type of \emph{inefficiency} \hhedit{in team compositions}: teams fail to bring on an adequate number of individuals with divergent predictive models to achieve the highest level of performance. }
Our results formalize a fundamental limitation of utility-based motivations to drive informational diversity in organizations and hint at interventions that may improve informational diversity and performance simultaneously.
\end{abstract}

\section{Introduction}

A long line of work in the social sciences has argued that, within organizational settings, groups that bring a diversity of perspectives to a task can be more effective than groups that do not.
The combination of distinct perspectives makes more information available
to a group, and can enable productive synergies among these sources
of information, improving a team's performance \citep{page2008difference, burt2004structural}.
Along with observations of this phenomenon in practice, a set of
mathematical models has sought to formalize these types of
informational advantages in abstract settings in which a group of
agents engage in collective problem-solving \citep{hong2004groups}.

If this form of \emph{informational diversity}\footnote{The literature sometimes refers to this type of diversity as \emph{cognitive} diversity~\citep{page2019diversity}. We use the term \emph{informational} diversity to emphasize that team members are bringing new informational resources to bear on solving problems.} confers performance advantages on teams
within organizations, why do we so often see teams that are largely
homogeneous in practice?
A canonical argument is that the benefits of informational diversity
are in tension with {\em affinity bias}, a human behavioral phenomenon
in which people prefer to interact with others who have similar perspectives.
This tendency is well-documented by prior work in organizational psychology~\citep{huang2019women,mccormick2015real,oberai2018unconscious}.
It is an aggregate effect that can stem from a range of different \hhedit{underlying}
mechanisms: for example, it may arise 
because people have an inherent preference for others
with similar perspectives, or because they have difficulty evaluating
others with different perspectives, or because they prefer teams with
fewer disagreements or those whose aggregate view is closer to their
own.  All of these would produce a version of affinity bias as an
outcome.
For purposes of our discussion here, we will focus on the observable
effects of these mechanisms in the form of affinity bias without restricting ourselves to a specific 
underlying mechanism.

The tension between informational diversity and affinity bias is
the basis for a number of empirical results establishing that 
informationally diverse teams can lead simultaneously to 
higher-quality solutions but also lower group cohesion
\citep{phillips2009pain,milliken1996searching,watson1993cultural}.
Such findings highlight the challenge in building informationally diverse 
teams: expanding a team by adding members with dramatically different
perspectives has the potential to improve the team's performance
but also to reduce the subjective value of the experience for participants
due to their affinity bias.
We are interested in understanding the \hhdelete{fundamental principles}\hhedit{the fundamental phenomena that emerge} from this conflict between informational diversity and affinity bias.
In particular, what are the implications of this tension for the composition of teams \hhedit{that form} as new members are brought on and the team grows in size?

\xhdr{\bf The present work: modeling informational diversity with affinity bias}
In this paper, we develop a model for team formation in the presence
of both informational diversity and affinity bias.
Whereas earlier models of informational diversity formulated
agent-level objective functions in such a way that the agents
should always favor greater diversity (see, e.g.,~\citep{lamberson2012optimal}),
our class of models explicitly captures the
tension between these forces in agents' utilities. 
In particular, in our model, 
agents \hhedit{forming} a team \hhedit{are} faced
with a prediction task: they see instances of a prediction
problem encoded by features, and they must make a prediction about
some future outcome for each instance.
A team of policymakers trying to predict the effect of 
policy interventions, 
a team of investors trying to predict which \hhedit{start-up} companies will be successful, 
a team of doctors faced with complex medical diagnosis, 
or a team of data scientists participating in web-based competitions such as the Netflix Prize and Kaggle competitions
are all among the types of scenarios captured by this framework.

While prior work assumes that agents aim to minimize the overall predictive error of their teams~\citep{lamberson2012optimal}, our approach uses these earlier formalisms as building blocks to produce a more general model in which each agent has an objective function \hhedit{comprised of the} sum of two terms
(capturing the two forces we are considering):
one term is the error rate of the team, and the other term
is their level of dissimilarity to other team members. 
A one-dimensional parameter controls the relative weight of these
two terms in the objective function; this general form for the objective function allows us to study 
extremes in which agents care primarily about team performance
or primarily about team homogeneity.

We are particularly interested in the process by which teams grow over
time, as they decide sequentially which new members to add.
For any configuration of a team, we can ask which types of agents
the team would be willing to add, where the criterion for adding
a new member is that it improves the aggregate utility of the
\hhdelete{current} team members (via the weighted sum of team performance and 
individual disagreement with others). Our main takeaway from the analysis of this model is a path-dependent characterization of inefficiencies in team compositions formed through the above sequential growth dynamics. 
This characterization hints at
organizational interventions that may improve informational diversity and performance simultaneously, including those that help 
reduce the impact of affinity bias on team formation dynamics,
or those that initiate teams from a more informationally diverse
composition (thereby beginning the dynamics at more favorable
initial conditions).

\vspace{-2mm}
\subsection{Model Overview}
We consider a setting in which a team consisting of multiple members is tasked with making complex, non-routine decisions based on the members' collective predictions about some future outcome.
Importantly,  the task is \emph{complex} and not further decomposable into \emph{specialized sub-tasks} that can be \hhedit{accomplished}\hhdelete{performed} independently\hhdelete{ and aggregated readily}.\footnote{\hhedit{Note that in our model, even though agent types rely on different sets of features to reach their predictions, each agent tries to solve the \emph{same} problem (i.e., predicting the outcome for a given state of the world) in its entirety.}} We have mentioned several examples of real-world settings in which these conditions are approximately met. As an example, consider aggregating the diverse forecasts of individual members of a marketing team to predict a new product's expected sales~\citep{lamberson2012optimal}.

\xhdr{Team problem-solving mechanism}
Team members have predictive models of the world (e.g.,  predicting the expected sales of a product).
Given a new case, each member uses their model to predict the outcome of the case. (In different contexts, this model can either be an abstraction of the team member's mental model of the domain, or it could be an actual implementation of a computational model that they have built.  Our model operates at a level of generality designed to address both these scenarios in general terms.)  
For simplicity, we assume individuals belong to one of the two opinion groups or \emph{types}, with those belonging to the same type holding similar predictive models. \hhedit{More precisely, an agent belonging to a given type has access to a \emph{noisy} version of that type's predictive model.}
\hhdelete{For now, let's assume t}the team aggregates its members' opinions into a collective prediction/decision using \hhedit{an aggregation mechanism, such as }simple averaging. 

\xhdr{Team's (dis)utility ($\lambda$)} The team's cost function combines two factors: (1) the expected error rate of the team's predictions; and (2) the dissimilarity among team members' predictors. In particular, the team aims to minimize:
$$\lambda \times \text{level of dissimilarity among team members} + (1-\lambda) \times \text{team's predictive error},$$
where the dissimilarity between two teammates is captured by the expected level of disagreement in their predictions for a randomly drawn case.
Note that various choices for $\lambda$ reflect different team preferences. For example, 
$\lambda= 0$ corresponds to a team which is solely concerned with improving accuracy.
$\lambda= 1$ corresponds to a team which only cares about minimizing internal disagreement.
Intermediate $\lambda$ values (e.g., $\lambda=0.5$) capture teams that weigh both accuracy and similarity.

\xhdr{The effect of team size ($\beta$)}
To capture the relationship between team size and level of dissimilarity among team members, we introduce a parameter $\beta \in [0, 1]$, which at a high-level captures the psychological costs of cooperation~\citep{borocs2010struggles} and managing conflicts~\citep{higashi1993determines} as the team grows in size. In particular, as described in Section~\ref{sec:model}, for a team of size $n$, we assume that the sum of pairwise disagreements among team members is normalized by $1/n^{1+\beta}$. 
This means that when $\beta = 0$ the psychological cost of disagreement grows linearly in the number of teammates who hold conflicting opinions, whereas when $\beta = 1$ the cost depends only on the fraction of agents with conflicting opinions, independent of team size.  Values of $\beta$ strictly between 0 and 1 interpolate between these extremes.

\xhdr{A \hhedit{parametric}\hhdelete{richer} class of aggregation mechanisms ($\alpha$)} \hhedit{While prior work takes simple averaging as the team's approach to aggregating predictions,} to better understand the effect of the aggregation mechanism \hhedit{on team's composition}, \hhdelete{instead of solely focusing on simple averaging of predictions,} we consider a natural class of aggregation functions parameterized by $\alpha \geq 0$. \hhedit{This parametric family is} akin to the Tullock's contest success function~\citep{jia2013contest,skaperdas1996contest}, in which a homogeneous subgroup of size $m$ in the team has its opinion weighted by $m^\alpha$ in the overall team aggregation.
In the case of only two opinion types, $\alpha = 1$ corresponds to the uniform average, and the limit $\alpha \rightarrow \infty$ corresponds to the majority rule (or the median).

\vspace{-2mm}
\subsection{Team Growth Dynamics}
We assume $\lambda$, $\beta$, and $\alpha$ are fixed throughout.
At each time step, a new agent arrives and the current team considers whether to bring them on as a new member. 
We assume this decision is made
by assessing whether the addition of the new agent to the team would reduce its disutility. \hhedit{Note that adding more than one member of each type may be desirable to the team for two distinct reasons: (a) since each agent’s prediction is a \emph{noisy} version of its type, adding multiple members of the same type leads to noise reduction; (b) having more than one member of each type might be necessary for the team to achieve the accuracy-optimal balance between types. (For example, if the accuracy-optimal composition consists of twice as many agents of type $A$ compared to type $B$, a team starting with 1 member of each type may find it beneficial to hire another member of type $A$.) }

For simplicity, \hhedit{in the basic version of our model, }we assume that teams can precisely measure both their internal levels of dissimilarity and their predictive errors. In particular, 
we make the simplifying assumption that a team can perfectly estimate its current accuracy as well as changes in accuracy as a result of bringing on a new member. This assumption is common in prior work (see, e.g., \citep{lamberson2012optimal,hong2020contributions}). 
Our key observation is that even with this optimistic assumption in place, teams fail to raise sufficient informational diversity to optimize accuracy. As discussed in Section~\ref{sec:extension}, if teams underestimate the accuracy gains of increased informational diversity, the incentive to bring on diverse team members would only be further hampered. 
(We will further demonstrate the effect of \hhedit{both the under- and over-estimation of accuracy gains on team growth dynamics} in Section~\ref{sec:extension}.)

\begin{figure}[h]
    \centering
    \includegraphics[width=0.8\textwidth]{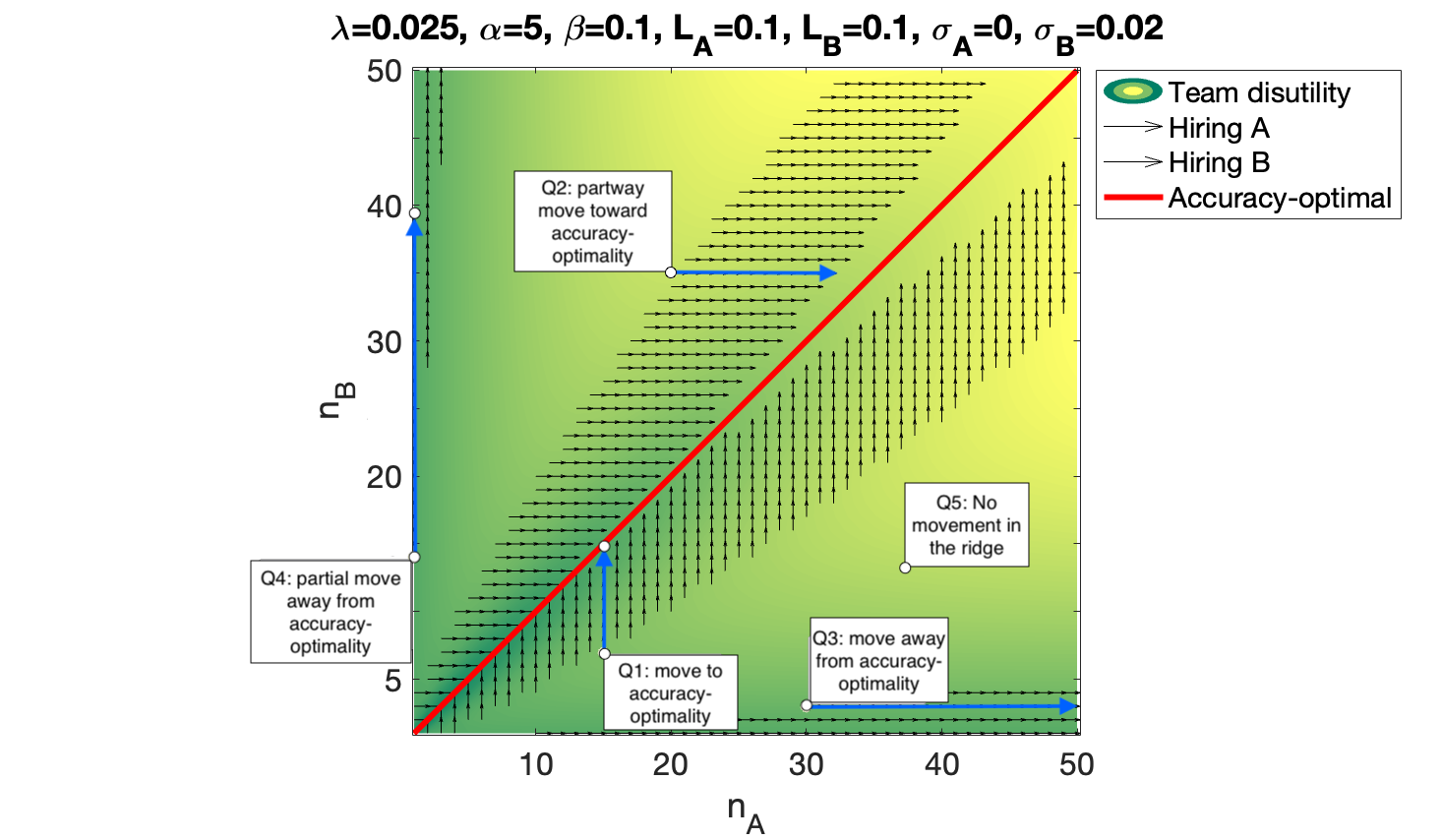}
    \vspace{-4mm}
    \caption{Preview of team growth dynamics. The x-axis specifies the number of type $A$ members  in the team and the y-axis, the number of type $B$ members. Arrows point in the direction of disutility reduction. The eventual composition is highly \emph{path-dependent} and often \emph{inefficient}.}
    \label{fig:preview}
    \vspace{-4mm}
\end{figure}

\vspace{-2mm}
\subsection{Insights from the Analysis}
Our analysis {characterizes} the kind of teams that form as the result of the interplay between predictive accuracy and affinity bias, depending on the three {primary} parameters of the model:
\begin{itemize}
\item \emph{$\lambda$, or the relative impact of affinity bias vs. accuracy on team growth dynamics:} In the extreme cases of $\lambda=0,1$ the team formation dynamics behave as one may expect: When the team only cares about accuracy ($\lambda=0$), it reaches the accuracy (=utility) optimal composition regardless of its initial makeup. If the team solely cares about reducing disagreement, only the initial majority type can bring on more members of its own. 
 For the intermediate values of $\lambda$, however, it is not a priori clear how inefficiency in team's accuracy emerges as $\lambda$ grows. One may expect a tipping point phenomenon, where $\lambda$ has to be larger than a certain threshold to hinder the formation of accuracy optimal teams.  With relatively few assumptions, our analysis shows that the pattern is, in fact, markedly different: For \emph{any} intermediate value of $\lambda$ team formation dynamics get stuck in local utility optima, failing to achieve accuracy-optimal compositions.
\item \emph{$\alpha$, or the mechanism by which individual predictions get aggregated into a team prediction:} 
As $\alpha$ grows, the team's rule for arriving at a collective prediction varies smoothly from pure averaging to a median or majority rule type of function.
In the process, the majority type's prediction has an increasingly dominant effect on the collective prediction as $\alpha$ grows.
This leads to a dynamic in which the prospect of adding new members from the 
less-represented type produces negligible accuracy gains but
non-trivial disagreement cost; as a result, new members from the 
less-represented group will not be added unless their relative size
on the current team is already substantial enough.
\item \emph{$\beta$, or the impact of team size on perceptions of within-team disagreements:} 
As $\beta$ becomes smaller, team size will play a more dominant role in affecting perceptions of dissimilarity/disagreement. 
As a result, the team will never add more than a certain number of the less-represented type. Depending on the initial composition of the team, the majority type may find it beneficial to continue adding new members of its own to drown out the predictions of the other type, and thereby drive down the cost
arising from dissimilarity. \hhedit{When within-type disagreements are non-zero---which can be the case due to the noise in the predictions made by agents of the same type---the majority type may stop expanding itself to avoid increasing within-type disagreements.}
\end{itemize}

Taken together, these principles suggest that team-growth dynamics can exhibit a set of subtle behaviors: (i) from certain initial team compositions, they can make progress toward better performance but then get stuck partway to optimally diverse teams; but (ii) from other initial compositions, they can also move away from this optimal balance as the majority group tries to crowd out the opinions of the minority. The initial composition of the team can determine whether the dynamics will move toward or away from performance optimality. 

It is natural to visualize this process geometrically as taking place in a Cartesian plane where the point $(n_A,n_B)$ represents a team with $n_A$ members of group $A$ and $n_B$ members of group $B$ \hhedit{and arrows initiating from point $(n_A,n_B)$ point to the direction in which growing the team would improve its utility}. Team growth dynamics then correspond to a walk through this space; and the \hhedit{destination}\hhdelete{direction} that this walk \hhedit{heads to}\hhdelete{takes} depends on the point it starts from. Figure~\ref{fig:preview} provides an example for how this analysis operates on a specific instance of the problem. In the figure, the diagonal line shows the optimal team composition, and arrows starting from a point $(n_A,n_B)$ on the plot indicate the direction in which a team consisting of $n_A$ members of groups $A$ and $n_B$ members of group $B$ grows. (The specific instance in the figure is described by parameter values $\alpha=5$, $\beta=0.1$, $\lambda=0.025$ using the notation from earlier; and both opinion types, $A$ and $B$, have the same error rate of $0.1$ ($\cL^A = \cL^B = 0.1$). \hhedit{As will be described in Section~\ref{sec:model}},\hhdelete{Recall that} in our model, we assume these similar error rates are achieved using different predictive attributes; therefore, teams consisting of both types achieve higher accuracy than homogeneous ones.) 

The figure illustrates in a concrete example the set of underlying principles that are formalized by our results. Specifically, looking at how the arrows for team growth point in different parts of the plane, we see that the space decomposes into a set of different regions with distinct behaviors. There is a {\em valley} near the diagonal: a subset of points close to the \hhedit{accuracy }optimum where growth dynamics will move the team toward\hhdelete{ the diagonal}. Some of these, like point $Q_1$ in the figure, will iterate all the way to \hhedit{accuracy} optimality, while others, like point $Q_2$ in the figure, will move partway to the optimum and then get stuck. There is also a {\em downslope} near each axis: points like $Q_3$ that are sufficiently close to the axis will actually move away from the \hhedit{accuracy} optimum and further out along the axis, corresponding to teams that add more of the majority type to reduce average dissimilarity. \hhedit{The expansion of the majority group through growth dynamic may stop if within-team disagreements become non-negligible (see, e.g., the dynamics initiating at point $Q_4$)}. Finally, there is a {\em ridge} that separates the central valley from the outer downslope; which side of the ridge a point is on determines whether it iterates in the direction of optimality or away from it. Points that actually lie on this ridge, like $Q_5$, do not move at all. 

Thus, our results suggest that there can be a critical level of team heterogeneity in the process: once the team passes this level of heterongeneity, then the growth dynamics will improve its performance; but if it falls short of this level of heterogeneity, then the growth dynamics may cause it to unravel toward greater homogeneity and lower performance. 

The type of analysis outlined above, while stylized in the context of our model, suggests several broader insights that can be actionable. First, the positive impact of diversity on a team's performance alone will not incentivize high-performing teams to form. Second, the analysis highlights some of the levers available to the planner to influence the team growth dynamics. Some of these are visible in Figure~\ref{fig:preview}, like the choice of initial team composition. Others are implicit in the choices of parameters --- for example, in the choice of aggregation mechanism (corresponding to $\alpha$) for resolving conflicts of opinion among team members. 

\vspace{-2mm}
\section{Related Work}
\xhdr{Optimal forecasting teams} Combining multiple predictors to achieve better predictive accuracy is a common and well-studied approach in machine learning, operations research, and economics \citep{clemen1989combining,armstrong2001principles}. Prior work has showed that combining a diverse set of predictors often improves performance \citep{batchelor1995forecaster,hong2004groups,lobo1990combining}. Motivated by the empirical evidence, prior work has proposed formal models of forecast aggregating teams~\citep{lamberson2012optimal,davis2015composition}. For example, \citet{lamberson2012optimal} focus on the role of team size on determining its optimal composition for making predictions. While our model closely follows \citep{lamberson2012optimal}, the question we are interested in is fundamentally different. It is also worth noting that the team formation process in our model can be viewed as a variant of ensemble learning in machine learning---with the key difference that there exists a cost associated with combining diverse models.

\xhdr{Comparison with \citep{lamberson2012optimal}}
Our work extends the model proposed by \citeauthor{lamberson2012optimal}, who study the optimal composition of teams making combined forecasts. Similar to our model, \hhedit{accuracy serves as a proxy for teams’ problem-solving abilities; the aggregation mechanism is fixed ahead of time;} agents belong to one of the two predictive types, $A$ and $B$, and a positive and fixed covariance exists between the errors made by any two agents of the same type. The key question is ``what composition of types minimizes the team's mean squared error?''. \citeauthor{lamberson2012optimal}'s key finding is that for large teams, the optimal composition is mainly comprised of the type with the lowest error covariance, even if the type is not the most accurate. In contrast, in small groups, the highest accuracy type will be in the majority.
Our major point of departure from \citep{lamberson2012optimal} is the team's objective function: Instead of assuming teams solely aim to maximize accuracy, we also account for the effect of affinity bias. Additionally, while \citeauthor{lamberson2012optimal}'s analysis focuses on uniform averaging of forecasts across team members, we study a richer class of aggregation mechanisms. Finally, unlike the prior contribution, which investigates the effect of within-type error covariance on accuracy-optimal compositions, we fix the error covariance of types and instead focus on teams' growth dynamics as the tensions between accuracy and affinity bias play out.

\xhdr{Diversity in team performance and dynamics} 
A substantial body of empirical and theoretical research has investigated the impact of diversity on teams' performance and dynamics. A significant part of this literature studies diversity with respect to demographic characteristics such as race, gender, and age/generation~\citep{guillaume2017harnessing, pelled1996demographic,elsass1997demographic}. Other scholars have focused on diversity in \emph{job-related}\footnote{According to \citet{pelled1996demographic}, ``job-relatedness is the extent to which the variable directly shapes perspectives and skills related to cognitive tasks.''} characteristics such as education level or tenure~\citep{sessa1995diversity,milliken1996searching,pelled1999exploring}. Our work is closer to the latter category of diversity. Importantly, our contributions do not directly apply to demographic diversity in organizations. 

Empirical work has investigated the impact of diversity on group performance and effectiveness. 
Some of the prior work argues that diversity can be a ``double-edged sword,'' meaning that it can lead to higher-quality solutions, while reducing group cohesion~\citep{phillips2009pain,milliken1996searching,lauretta1992effects,watson1993cultural,o1989work}. The goal of our analysis is to understand why and under what conditions diversity acts this way.

%
Aside from performance, empirical studies have established that groups consisting of dissimilar individuals leads to less attraction and trust among peers \citep{chattopadhyay1999beyond}, less frequent communication \citep{zenger1989organizational}, lower group commitment and psychological attachment \citep{tsui1992being}, lower task contributions \citep{kirchmeyer1993multicultural,kirchmeyer1992multicultural}, and lower perceptions of organizational fairness and inclusiveness \citep{mor1998organizational}. 
Compared to homogeneous groups, heterogeneous groups are found to have reduced cohesiveness \citep{terborg1975longitudinal}, more conflicts and misunderstandings \citep{jehn1997agree} which, in turn, lowers members' satisfaction, decreases cooperation \citep{chatman2001influence}, and increases turnover \citep{jackson1991some}. These empirical findings are reflected through an inherent taste for agreement among team members in our model.

\xhdr{Affinity bias} Affinity bias or homophily is the tendency of individuals to gravitate toward or associate with others whom they consider similar to themselves. The similarity could be in terms of demographic characteristics (such as race, ethnicity, age, or gender), social status (e.g., job title), values (e.g., political affiliation), or beliefs. A substantial body of empirical work has established the existence of homophilly in social networks
~\citep{mcpherson2001birds}. Affinity bias in organizational processes has been documented and discussed extensively~\citep{huang2019women,mccormick2015real,oberai2018unconscious}.

\xhdr{Hedonic games} Hedonic games model the formation of coalitions (or teams) of players in settings where players have preferences over coalitions~\citep{bogomolnaia2002stability}. Existing work in the area focuses on the stability of game outcomes (e.g., by evaluating whether the outcome of the game belongs to the core).
Similar to hedonic games, in our setting, each agent's payoff depends on the other members of her team. However, unlike hedonic games, we are not interested in how society partitions itself into disjoint coalitions. Instead, we study a team that evolve sequentially when current members get to decide who joins next.

\xhdr{Wisdom of crowds and prediction markets} The wisdom of crowds~\citep{surowiecki2005wisdom,mannes2012social} capture the idea that groups of people often perform better at prediction tasks compared to individuals. This idea has been the basis of ``prediction markets'' where agents can buy or sell securities whose payoff correspond to future events. The market prices can indicate the crowd's collective belief about the probability of the event of interest~\citep{wolfers2004prediction,arrow2008promise}. In prediction markets, traders do not generally form groups or coalitions; rather, they bet against each other. A trader receives the highest possible payoffs only if their prediction about the future state of the world is correct and only a small subset of other traders have made their bet according to the correct prediction. Unlike prediction markets, our model assumes individual members of a team are concerned with the \emph{overall} performance of their team. Additionally, we set aside incentive considerations to focus on the interplay between accuracy and affinity bias in team growth dynamics.

\vspace{-1mm}
\section{The Basic Model}\label{sec:model}
Let $\cX$ denote the set of all possible states of the world distributed according to a probability distribution $P$. We assume each state of the world is described by a feature vector, $\vx = (x_1, \cdots, x_r) \in \cX$, consisting of uncorrelated attributes $x_1, \cdots, x_r$, that is, $\text{cov}(x_i, x_j) = 0$ for all $j \neq i$. Each state of the world, $\vx$, leads to an outcome $y \in \cY$. We assume there exists a true outcome function $f^*$, such that for any $\vx \in \cX$, $y = f^*(\vx)$ is the true outcome of $\vx$. For simplicity and unless otherwise specified, we assume $f^*$ is deterministic, $\cX = \reals^r$, and $\cY = \reals$.

Consider a set of agents all capable of making predictions about the true outcome given the state of the world.  
An agent $i$ has a \emph{fixed} predictive model of the world, denoted by $f_i : \cX \longrightarrow \cY$, which maps each possible state of the world, $\vx \in \cX$, to a predicted outcome, $\hy_i = f_i(\vx)$. We will use $\cL_i$ to denote the accuracy loss of $i$'s predictions. More precisely, given a loss function $\ell: \cY \times \cY \longrightarrow \reals$,
$$\cL_ i = \Exp_{\vx \sim P} \left[ \ell(\hy_i, y) \right].$$
As an example, $\ell$ can be the squared error, that is, $\ell(\hy_i, y) = (\hy_i - y)^2$.

For any two predictive models $f_i,f_j$, we define the level of disagreement between them through a distance metric, $\delta: \cY \times \cY \longrightarrow \reals^+$. 
$$d_{i,j} = \Exp_{\vx \sim P} \left[ \delta(f_i(\vx), f_j(\vx)) \right].$$
As an example, $\delta$ can be the squared $L_2$ norm, that is, $\delta(f_i(\vx), f_j(\vx)) = (f_i(\vx)- f_j(\vx))^2$.
To simplify the analysis, we will first focus on simple quadratic loss ($\ell$) and distance ($\delta$) functions, i.e., $\ell(y,y') = \delta(y,y') = (y-y')^2$. Later in Section~\ref{sec:extension}, we show that our results extend to a larger family of distance metrics and loss functions.

\xhdr{Agent types} For simplicity and following prior work~\citep{lamberson2012optimal}, we assume there are two \emph{types} of agents, $A$ and $B$, each with a type-specific predictive model of the world. In particular, individuals of each type base their predictions on a type-specific subset of features, and their predictions are a noisy version of the highest accuracy predictor on those features. (The assumption of individuals utilizing the highest-accuracy predictor available to them is common in prior work. See, e.g., ~\citep{hong2020contributions}.) \emph{Without loss of generality},\footnote{Suppose the intersection of the two types' feature sets is non-empty. Since we assume both teams train the highest-accuracy predictor on their feature set, they will weigh the common features similarly. By subtracting the portion of their predictors corresponding to common features from $f^*$, $f_A$ and $f_B$, we can equivalently assume types have access to disjoint features.} suppose an individual of type $A$ only takes features $x_1, \cdots, x_k$ into consideration, whereas an individual of type $B$ utilizes features $x_{k+1} ,\cdots , x_r$ to make a prediction about $\vx$. 

We assume the true function $f^*$ is linear in the feature vector $\vx$. Therefore, it can be decomposed into two components---corresponding to the two types' feature sets. In particular, $\forall \vx=(x_1,\cdots,x_r) \in \cX:$
\begin{eqnarray}\label{eq:team_cost}
f^*(\vx) &&=  \vtheta^* \cdot \vx  =  \theta^*_1 x_1   + \cdots + \theta^*_k x_k  + \theta^*_{k+1} x_{k+1}  + \cdots + \theta^*_r x_r \\
&&= \vtheta^*_{1, \cdots, k} \cdot (x_1, \cdots , x_k) + \vtheta^*_{k+1, \cdots, r} \cdot (x_{k+1} ,\cdots , x_r)
\end{eqnarray}
We will refer to $\vtheta^*_{1, \cdots, k}$ as $\vtheta^A$ (since it's the accuracy-optimal weights on $A$'s features) and $\vtheta^*_{k+1, \cdots, r}$ as $\vtheta^B$, so that $f^*(\vx) = \vtheta^A \cdot \vx + \vtheta^B \cdot \vx$. Given an instance $\vx$, individuals of each type can produce a noisy prediction according to their type's highest-accuracy predictive model. More specifically, an individual of type $A$ predicts 
$$f^A(\vx) = \vtheta^A \cdot \vx + \theta^A_0+ \epsilon_A$$ for $\vx$, where $\epsilon_A$ is an i.i.d. noise sampled from a mean-zero Gaussian distribution with variance $\sigma^2_A$, and $\theta^A_0 = \Exp[\vtheta^B \cdot \vx]$ is the intercept. Similarly, an individual of type $B$ predicts $f^B(\vx) = \vtheta^B \cdot \vx + \theta^B_0 + \epsilon_B$. 
To avoid having to carry the dot-product notation, we define the shorthand functions $\theta^A(\vx) := \vtheta^A \cdot \vx$ and $\theta^B(\vx) := \vtheta^B \cdot \vx$. We will also use $\cL^A, \cL^B$ to refer to the (noise-less) accuracy loss of each type's predictive model (i.e., $\cL^c = \Exp_{\vx \sim P} \left[ \ell(\theta^c \cdot \vx, y) \right]$ for $c \in\{A,B\}$).
Additionally, for simplicity, we assume $\Exp[\vx] = \zeros$\footnote{We can ensure this condition by standardizing all features.} so that the above intercepts are both 0. (Note that since $\Exp [\vx] = \zeros$, and $f^*, f^A$, and $f^B$ are all linear in $\vx$, we have that $\Exp_{\vx \sim P} [f^*(\vx)] = \Exp_{\vx \sim P} [f^A(\vx)] = \Exp_{\vx \sim P} [f^B(\vx)] = 0$). 

\xhdr{The aggregation mechanism} A team $T$ is a set of agents who combine their predictions according to a given aggregation function $\cG_T$. For any $\vx \in \cX$, the aggregation function $\cG_T$ receives the predictions made for input $\vx$ by all members of $T$, and output a collective/team prediction for $\vx$. 
Given the team $T=\{1,2, \cdots, |T|\}$, we use $\cG_T(\vx)$ to refer to the aggregated prediction of team members for state $\vx$. More precisely,
$$\cG_T(\vx) = \cG(f_1(\vx), \cdots, f_{|T|}(\vx)).$$
Note that we assume the functional form of $\cG_T$---denoted by $\cG$ in the above equation---is exogenously chosen and fixed throughout, and is independent of the team's composition.
As a concrete example, $\cG$ can be the mean prediction across all team members, that is:
$$\cG(f_1(\vx), \cdots, f_{n}(\vx))= \frac{1}{n} \sum_{i=1}^n f_i(\vx) = \frac{1}{n} \sum_{i=1}^n \hy_i .$$
(This particular form of $\cG$ is reasonable to assume in environments where the only acceptable aggregation function is giving all team members' opinions equal weight.)
We will consider a general class of aggregation functions inspired by Tullock's contest success function~\citep{jia2013contest,skaperdas1996contest}, defined as follows:
Given a team consisting of $n_A$ individuals of type $A$ and $n_B$ individuals of type $B$, we define the following parametric class of aggregation functions:
\begin{equation}\label{eq:tullock}
 \forall \alpha \in [0, \infty): \indent \cG^\alpha_{n_A,n_B}(\vx) = \left(\frac{n^\alpha_A}{n^\alpha_A + n^\alpha_B}\right)  f^A(\vx) + \left(\frac{n^\alpha_B}{n^\alpha_A + n^\alpha_B}\right) f^B(\vx).  
\end{equation}
When $n_A, n_B$ are clear from the context, we drop the subscript and use $\cG^\alpha$ to refer to the aggregation function. 
Note that if $\alpha=1$, the above simplifies to the simple average, and in the limit of $\alpha \rightarrow \infty$, $\cG^\alpha$ becomes close the median. 

\xhdr{Disutility and cost}
Individual agents can be added to a team to reduce the team's overall \emph{disutility} or \emph{cost}. 
The cost an agent $i$ incurs as a member of team $T$ is defined as a combination of (a) their level of disagreement with other team members, and (b) the team's overall accuracy loss. More precisely,
\vspace{-2mm}
{
$$c_i(T) =  \lambda \times \frac{1}{|T|^\beta}\sum_{j \in T}  \Exp_{\vx \sim P}\left[ \delta(f_j(\vx), f_i(\vx)) \right]  + (1-\lambda) \times \Exp_{\vx \sim P}\left[ \ell(\cG_{T}(\vx), y) \right].$$
}
Note that the parameter $\beta \in [0,1]$ specifies how the level of disagreement experienced by individual $i$ is impacted by the size of the team. In particular, it captures how perceptions of disagreement scale with absolute vs. relative size of the opposing type. As an example to illustrate this, consider a hypothetical case where $\Exp_{\vx \sim P}\left[ \delta(f_j(\vx), f_i(\vx)) \right]$ is fixed to some constant value $\delta$ for all $j \neq i$. When $\beta=0$, $i$'s experienced disagreement with team members is equal to $(|T-1|\delta)$, and it grows linearly with team size $|T|$. That is, a larger team increases $i$'s perception of disagreement with teammates. For $\beta=1$, though, $i$'s experienced disagreement level ($|T-1|\delta/|T|$) remains roughly constant. 

\xhdr{Growth dynamics}
We assume 
a team $T$ would be willing to accept new member if it reduces the team's average cost across its current members:
{
\begin{equation}\label{eq:general_disutility}
    \lambda \times \frac{1}{|T|^{1+\beta}}\sum_{i,j \in T}  \Exp_{\vx \sim P}\left[ \delta(f_j(\vx), f_i(\vx)) \right]  + (1-\lambda) \times \Exp_{\vx \sim P}\left[ \ell(\cG_{T}(\vx), y) \right].
\end{equation}
}

The team growth dynamics proceeds as follows: Team $T$ initially consists of $n_A$ individuals of type $A$ and $n_B$ individuals of type $B$. New individual candidates arrive one at a time over steps $t=1,2,\cdots$. Let us refer to the $t$'th individual as $i_t$, and denote their type by $s_t \in \{A,B\}$. The team brings on $i_t$ if and only if it reduces the current team's disutility (\ref{eq:general_disutility}). 
\hhedit{As will be discussed in Section~\ref{sec:analysis}, there are two distinct reasons for the team to bring on more than one member of a given type: 
\begin{itemize}
    \item When the accuracy-optimal team composition favors one type, hiring more than one member of each type might be necessary to achieve the accuracy-optimal balance. 
    \item When each agent’s prediction is a noisy version of its type, multiple members of the same type reduces noise. 
\end{itemize}
}

\section{Analysis}\label{sec:analysis}
In this section, we describe the dynamics of team formation under three separate regimes of $\lambda$: one in which $\lambda=0$ (only accuracy matters); another in which $\lambda$ is close to $1$ (the importance of accuracy is negligible), and finally settings in between where $0 < \lambda <1$.

\subsection{Minimum $\lambda$ value ($\lambda = 0$)}\label{sec:lambda_0}
When $\lambda = 0$, the team adds new members if and only if the new member reduces the team's mean squared error. To see under what conditions a new member will be brought on, suppose the current team consists of $n_A$ individuals of type $A$ and $n_B$ individuals of type $B$. 
Recall that the team's collective predictive model can be written as 
$
\cG^\alpha(\vx) = \frac{n^\alpha_A}{n^\alpha_A + n^\alpha_B} \theta^A(\vx) + \frac{n^\alpha_B}{n^\alpha_A + n^\alpha_B} \theta^B(\vx) +  \frac{n^\alpha_A \epsilon_A + n^\alpha_B\epsilon_B }{n^\alpha_A+n^\alpha_B}.
$
It is easy to show that the team's mean-squared error can be decomposed into bias and variance terms.
\hhedit{
\begin{lemma}[Team's Error Decomposition]\label{lem:accuracy_decomp} 
Consider a team with a composition of $n_A$ members of type $A$ and $n_B$ members of type $B$. Then:
\begin{equation}\label{eq:decomp_lambda0}
\Exp \left[ \left( \cG_T(\vx) - y \right)^2 \right]  =  
\frac{n_B^{2\alpha}}{(n^\alpha_A + n^\alpha_B)^2}  \cL^B + \frac{n_A^{2\alpha}}{(n^\alpha_A + n^\alpha_B)^2}   \cL^A + \frac{n^{2\alpha}_A \sigma^2_A + n^{2\alpha}_B \sigma^2_B}{(n^\alpha_A + n^\alpha_B)^2},
\end{equation}
where the expectation is with respect to $(\vx,y) \sim P$ and $\epsilon_c \sim \cN(0,\sigma^2_c)$ for $c \in \{A,B\}$.
\end{lemma}
}
\begin{proof}
We can write:
{\footnotesize
\begin{eqnarray*}
&&\Exp \left[ \left( \cG_T(\vx) - y \right)^2 \right]  =  \Exp \left[ \left( \cG_T(\vx) - f^*(\vx) \right)^2 \right] \\
&= & \Exp \left[ \left(\frac{n^\alpha_A}{n^\alpha_A + n^\alpha_B} \theta^A(\vx) + \frac{n^\alpha_B}{n^\alpha_A + n^\alpha_B} \theta^B(\vx) + \frac{n^\alpha_A \epsilon_A + n^\alpha_B\epsilon_B }{n^\alpha_A+n^\alpha_B} - f^*(\vx) \right)^2 \right] \\
&= & \Exp \left[ \left(\frac{-n^\alpha_B}{n^\alpha_A + n^\alpha_B} \theta^A(\vx) + \frac{-n^\alpha_A}{n^\alpha_A + n^\alpha_B} \theta^B(\vx) + \frac{n^\alpha_A \epsilon_A + n^\alpha_B\epsilon_B }{n^\alpha_A+n^\alpha_B}\right)^2 \right] \indent \text{(Using \ref{eq:team_cost})}\\
&=& \frac{n_B^{2\alpha}}{(n^\alpha_A + n^\alpha_B)^2}  \Exp \left[  \theta^A(\vx)^2  \right] + \frac{n_A^{2\alpha}}{(n^\alpha_A + n^\alpha_B)^2}   \Exp \left[  \theta^B(\vx)^2  \right] + \frac{2 n^\alpha_A n^\alpha_B}{(n^\alpha_A + n^\alpha_B)^2}  \Exp \left[  \theta^A(\vx) \theta^B(\vx)   \right] 
 + \frac{1}{(n^\alpha_A + n^\alpha_B)^2} \Exp[(n^\alpha_A\epsilon_A + n^\alpha_B\epsilon_B)^2]
\end{eqnarray*}
}
Next, using the fact that the two types don't have access to common features, and the fact that $\text{cov}(x_i, x_j) = 0$ for all $j \neq i$, we know $\Exp \left[  \theta^A(\vx) \theta^B(\vx)   \right] = 0$. Additionally, $\Exp \left[\epsilon_A \epsilon_B \right] = 0$ due to the assumption of independent noises. Therefore, the above equation can be simplified to:
{\footnotesize
\begin{eqnarray*}
&=& \frac{n_B^{2\alpha}}{(n^\alpha_A + n^\alpha_B)^2}  \Exp \left[  \left(f^*(\vx) - \theta^B(\vx) \right)^2  \right] + \frac{n_A^{2\alpha}}{(n^\alpha_A + n^\alpha_B)^2}   \Exp \left[  \left(f^*(\vx) - \theta^A(\vx) \right)^2   \right] 
+ \frac{n^{2\alpha}_A}{(n^\alpha_A + n^\alpha_B)^2} \Exp[\epsilon_A^2] + \frac{n^{2\alpha}_B}{(n^\alpha_A + n^\alpha_B)^2} \Exp[\epsilon_B^2] \\
&=& \frac{n_B^{2\alpha}}{(n^\alpha_A + n^\alpha_B)^2}  \cL^B + \frac{n_A^{2\alpha}}{(n^\alpha_A + n^\alpha_B)^2}   \cL^A + \frac{n^{2\alpha}_A \sigma^2_A + n^{2\alpha}_B \sigma^2_B}{(n^\alpha_A + n^\alpha_B)^2}.
\end{eqnarray*}
}
\end{proof}
A team consisting of $(n_A, n_B)$ members of $A,B$ respectively will only add a new member of type $A$ if (\ref{eq:decomp_lambda0}) evaluated at $(n_A+1, n_B)$ is smaller than that at $(n_A, n_B)$. A similar logic applies to adding a new team member of type $B$.
Fixing the number of team members belonging to one type (say $A$), the following proposition specifies the accuracy-optimal number of members belonging to the other group (here, $B$).
\begin{proposition}[Accuracy-optimal composition]\label{prop:accuracy_optimal}
Consider a team with an initial composition of $n_A>0$ members of type $A$ and no member of type $B$. 
The optimal number of type $B$ members whose addition minimizes the team's accuracy loss in (\ref{eq:decomp_lambda0}) is equal to $n^*_B = n_A \left(\frac{\cL^A + \sigma_A^2}{\cL^B + \sigma_B^2}\right)^{1/\alpha}$.
\end{proposition}
\begin{proof}
According to Lemma~\ref{lem:accuracy_decomp}, the team's accuracy can be written as follows:
{\footnotesize
\begin{equation*}
\Exp \left[ \left( \cG_T(\vx) - y \right)^2 \right] = \frac{n_B^{2\alpha}}{(n^\alpha_A + n^\alpha_B)^2}  (\cL^B + \sigma^2_B) + \frac{n_A^{2\alpha}}{(n^\alpha_A + n^\alpha_B)^2}   (\cL^A+\sigma^2_A).
\end{equation*}
}
Taking the derivative of the right hand side with respect to $n_B$, we obtain:
{\footnotesize
\begin{eqnarray}
&& \frac{2 \alpha n^{2\alpha-1}_B(n^\alpha_A + n^\alpha_B)^2 - 2 \alpha n^{3\alpha - 1}_B  (n^\alpha_A+n^\alpha_B)}{(n^\alpha_A + n^\alpha_B)^4}  (\cL^B + \sigma^2_B) 
+ \frac{-2 \alpha(n^\alpha_A + n^\alpha_B) n_A^{2\alpha} n_B^{\alpha-1}}{(n^\alpha_A + n^\alpha_B)^4}   (\cL^A+\sigma^2_A) \nonumber \\
&=&  \frac{2 \alpha n^{\alpha-1}_B}{(n^\alpha_A + n^\alpha_B)^3}  \left( \left( n^{\alpha}_B(n^\alpha_A + n^\alpha_B) - n^{2\alpha}_B \right)  (\cL^B + \sigma^2_B) -  n^{2\alpha}_A (\cL^A+\sigma^2_A) \right) \nonumber \\
&=&  \frac{2 \alpha n^{\alpha-1}_B n^\alpha_A}{(n^\alpha_A + n^\alpha_B)^3}  \left( n^\alpha_B(\cL^B + \sigma^2_B) -  n^\alpha_A (\cL^A+\sigma^2_A) \right).\label{eq:derivative_lambda0}
\end{eqnarray}
}
To obtain the zero of the derivative, we can write:
{\footnotesize
$$n^\alpha_B(\cL^B+\sigma^2_B) -  n^\alpha_A (\cL^A+\sigma^2_A) = 0  \Leftrightarrow   n_B =  n_A   \left(\frac{\cL^A+\sigma^2_A}{\cL^B + \sigma^2_B}\right)^{1/\alpha}.
$$
}
\hhedit{It is trivial to see that for $n_B <  n_A   \left(\frac{\cL^A+\sigma^2_A}{\cL^B + \sigma^2_B}\right)^{1/\alpha}$, the expression above (\ref{eq:derivative_lambda0}) is negative, and for $n_B >  n_A   \left(\frac{\cL^A+\sigma^2_A}{\cL^B + \sigma^2_B}\right)^{1/\alpha}$ it is positive. Therefore, $n_B =  n_A   \left(\frac{\cL^A+\sigma^2_A}{\cL^B + \sigma^2_B}\right)^{1/\alpha}$ is indeed the minimum of the accuracy loss.
} This finishes the proof. 
\end{proof}
The Appendix contains several illustrations of team growth dynamics for $\lambda=0$ under several different regimes of $\alpha$ and $(\cL^A, \cL^B)$. Figures~\ref{fig:plain_lambda_0_alpha_1}, \ref{fig:plain_lambda_0_alpha_1_noise} demonstrate that regardless of the initial composition of the team, the growth dynamics converge to the accuracy/utility-optimal compositions. \hhedit{Additionally, we observe that it is never simultaneously beneficial for a team to add a member type $A$ and type $B$. These trends remain unchanged even if agents' predictions are noisy.}  
\hhedit{
\begin{remark}[The effect of $\sigma_A, \sigma_B$ on growth dynamics]
   Note that due to the symmetry of Equation \ref{eq:decomp_lambda0} in $A$ and $B$, the partial derivatives of the team's accuracy with respect to $n_A$ and $n_B$ always have opposing signs, therefore, at any $n_A, n_B \geq 0$, it is either beneficial to add a new member of type A or a new member of type $B$, but never both.  
\end{remark}
}

\subsection{Maximum $\lambda$ value ($\lambda = 1$)}\label{sec:lambda_1}
When $\lambda = 1$, the team's disutility simply corresponds to the disagreement term among team members. 
The following Lemma shows that the rate of disagreement can be decomposed.
\hhedit{
\begin{lemma}[Team's Disagreement Decomposition]\label{lem:distance_decomp}
Consider a team consisting of $n_A$ individuals of type $A$ and $n_B$ individuals of type $B$. Then
{\footnotesize
\begin{eqnarray} \label{eq:disutility_lambda1}
\Exp \left[ \frac{1}{(n_A + n_B)^{1+\beta}} \sum_{i,j \in T} d(i,j) \right] &=&  \frac{2n_A n_B}{(n_A + n_B)^{1+\beta}} \left( \cL^A + \cL^B + \sigma_A^2 + \sigma_B^2\right)+\nonumber \\
&+& \frac{2n_A (n_A -1)}{(n_A + n_B)^{1+\beta}} \sigma_A^2 + \frac{2n_B (n_B-1)}{(n_A + n_B)^{1+\beta}} \sigma_B^2,
\end{eqnarray}
}
where the expectation is with respect to $(\vx,y) \sim P$ and $\epsilon_c, \epsilon'_c \sim \cN(0,\sigma^2_c)$ for $c \in \{A,B\}$.
\end{lemma}
}
\begin{proof}
We can write the left hand side of (\ref{eq:disutility_lambda1}) as follows:
{\footnotesize
\begin{eqnarray*}
&&\frac{1}{(n_A + n_B)^{1+\beta}} \sum_{i,j \in T} \Exp \left[d(i,j) \right]= \frac{1}{(n_A + n_B)^{1+\beta}} \sum_{i,j \in T} \Exp \left[ (\theta_i(\vx) + \epsilon_i - \theta_j(\vx) - \epsilon_j)^2 \right]\\
&=&  \frac{2n_A n_B}{(n_A + n_B)^{1+\beta}} \Exp\left[ (\theta^A(\vx) + \epsilon^A - \theta^B(\vx) - \epsilon^B)^2 \right] + \frac{n_A (n_A-1)}{(n_A + n_B)^{1+\beta}} \Exp\left[ (\epsilon_A - \epsilon'_A)^2 \right] \\
&& + \frac{n_B (n_B-1)}{(n_A + n_B)^{1+\beta}} \Exp\left[ (\epsilon_B - \epsilon'_B)^2 \right] \\
&=&  \frac{2n_A n_B}{(n_A + n_B)^{1+\beta}} \Exp\left[(\theta^A(\vx)- \theta^B(\vx))^2 \right] + \frac{2n_A n_B}{(n_A + n_B)^{1+\beta}} \Exp\left[ (\epsilon^A - \epsilon^B)^2 \right] \\
&&+ \frac{n_A (n_A-1)}{(n_A + n_B)^{1+\beta}} \Exp\left[ (\epsilon_A - \epsilon'_A)^2 \right] + \frac{n_B (n_B-1)}{(n_A + n_B)^{1+\beta}} \Exp\left[ (\epsilon_B - \epsilon'_B)^2 \right]\\
&=&  \frac{2n_A n_B}{(n_A + n_B)^{1+\beta}} \Exp\left[(\theta^A(\vx)- \theta^B(\vx))^2 \right]  \\
&& + \frac{ 2n_A n_B + 2n_A (n_A-1)}{(n_A + n_B)^{1+\beta}} \Exp\left[ \epsilon_A^2 \right] + \frac{ 2n_A n_B + 2n_B (n_B-1)}{(n_A + n_B)^{1+\beta}} \Exp\left[ \epsilon_B^2 \right]\\
&=& \frac{2n_A n_B}{(n_A + n_B)^{1+\beta}} \left( \cL^A + \cL^B + \sigma_A^2 + \sigma_B^2\right) + \frac{2n_A (n_A -1)}{(n_A + n_B)^{1+\beta}} \sigma_A^2 + \frac{2n_B (n_B-1)}{(n_A + n_B)^{1+\beta}} \sigma_B^2.
\end{eqnarray*}
}
\end{proof}

If $\sigma_A = \sigma_B = 0$, Equation \ref{eq:disutility_lambda1} simplifies to the following:
{\footnotesize
\begin{equation}\label{eq:disutility_lambda1_noiseless}
\frac{2n_A n_B}{(n_A + n_B)^{1+\beta}} \left( \cL^A + \cL^B \right)
\end{equation}
}
To characterize the conditions under which adding a new member of type $B$ will be beneficial to the team, we can inspect the derivative. The derivative of (\ref{eq:disutility_lambda1_noiseless}) with respect to $n_B$ is equal to 
{\footnotesize
$$
2\frac{n_A (n_A + n_B)^{1+\beta} - (1+\beta)n_A n_B(n_A + n_B)^\beta}{(n_A + n_B)^{2+2\beta}} \left( \cL^A + \cL^B \right) = 
2\frac{n_A(n_A - \beta n_B)}{(n_A + n_B)^{2+\beta}} \left( \cL^A + \cL^B \right).$$
}
Note that for any $\beta \geq 0$, the derivative amounts to zero at $n_B = \frac{n_A}{\beta}$. Additionally, the derivative is strictly positive for $n_B < \frac{n_A}{\beta}$ and is strictly negative for $n_B > \frac{n_A}{\beta}$). This implies that adding a type $B$ agent to the team is beneficial if and only if $n_B \geq \frac{n_A}{\beta}$. 
In the Appendix, we illustrate team growth dynamics for $\lambda=1$, $\sigma_A=\sigma_B=0$ and several different values of $\beta$.

\hhedit{
The noisy case is slightly more nuanced. As formalized in the following Proposition, if the noisy type is in the majority (e.g., type $B$ is the majority in all above-the-diagonal compositions demonstrated in Figure~\ref{fig:plain_lambda_1_noisy}), adding more $B$-type members allows the majority to drowns out disagreements by $A$-type members. But after a while when when the disagreement with $A$-types has been suppressed sufficiently, it is no longer beneficial to add $B$ members because of the impact of noise on disagreements among same-type members.
\begin{proposition}
Consider a team with an initial composition of $n_A>0$ members of type $A$ and $n_B$ members of type $B$. Suppose $\sigma_B >0$ and $\beta<1$. Then there exist $n^\text{lower}_B, n^\text{upper}_B \in \mathbb{R}$ such that adding a type $B$ member reduces the team's disagreement in (\ref{eq:disutility_lambda1}) if and only if $n^\text{lower} \leq n_B \leq n^\text{upper}$.
\end{proposition}
\begin{proof}
The derivative of (\ref{eq:disutility_lambda1}) with respect of $n_B$ is equal to 
{\footnotesize
\begin{eqnarray*}
&& 2\frac{n_A(n_A - \beta n_B)}{(n_A + n_B)^{2+\beta}} \left( \cL^A + \cL^B + \sigma_A^2 + \sigma_B^2\right) \\
&-& 2\sigma_A^2 (1+\beta) \frac{n_A(n_A-1)}{(n_A + n_B)^{2+\beta}}\\
&+& 2\sigma_B^2 \frac{(1-\beta) n_B^2 + n_B(2 n_A +\beta) - n_A}{(n_A + n_B)^{2+\beta}}
\end{eqnarray*}
}
Note that setting the derivative to zero, is equivalent to solving the roots of the following function:
{\footnotesize
\begin{eqnarray*}
0 &=& 2n_A(n_A - \beta n_B) \left( \cL^A + \cL^B + \sigma_A^2 + \sigma_B^2\right)
- 2\sigma_A^2 (1+\beta) n_A(n_A-1)
+ 2\sigma_B^2 (1-\beta) \left(n_B^2 + n_B(2 n_A +\beta) - n_A \right)\\
&=& 2\sigma_B^2(1-\beta) n_B^2 \\
&&+ \left( -2\beta n_A\left( \cL^A + \cL^B + \sigma_A^2 + \sigma_B^2\right)  + 2\sigma_B^2 (1-\beta) (2n_A + \beta)\right) n_B \\
&&+ \left( 2n_A^2 \left( \cL^A + \cL^B + \sigma_A^2 + \sigma_B^2\right) - 2\sigma_A^2 (1+\beta) n_A(n_A-1) - 2\sigma_B^2 (1-\beta) n_A\right)
\end{eqnarray*}
}
Note that since $\sigma_B > 0$ and $(1-\beta)>0$, the above is a quadratic polynomial in $n_B$ with a positive leading coefficient (i.e., $2\sigma_B^2(1-\beta)$). Let $n^\text{lower}, n^\text{upper}$ denote the roots of this polynomial. Since the leading coefficient is positive, for any $n_B \in [n^\text{lower}, n^\text{upper}]$, the derivative of the disagreement term is negative, indicating that adding new members of type $B$ will reduce the disagreement. Similarly, outside this range, the derivative is positive indicating that new type $B$ members will only worsen the team's disagreement. This finishes the proof.
\end{proof}
}

\begin{figure}[h!]
    \centering
     \includegraphics[width=0.45\textwidth]{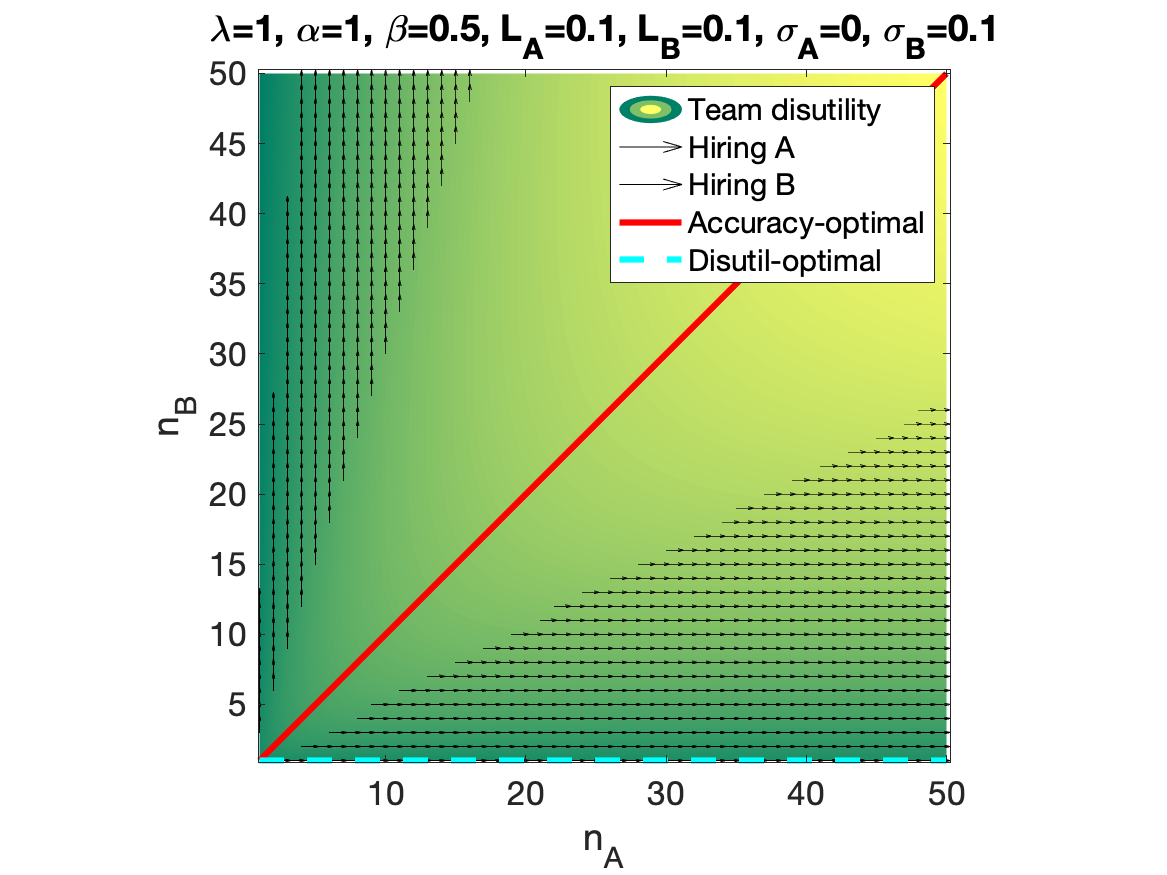}
    \caption{Visualization of team formation dynamics for $\lambda=1$ when $ \sigma_A =0, \sigma_B>0$. Adding a new member of  type $B$ is only beneficial in a specific range of $n_B$ ($n_B \in [n^\text{lower}, n^\text{upper}]$) determined by the roots of a degree-two polynomial. 
    }\label{fig:plain_lambda_1_noisy}
    \vspace{-4mm}
\end{figure}

\subsection{Intermediate $\lambda$ values}\label{sec:lambda_mid}
For $0 < \lambda < 1$, the team's disutility can be written as:
{\footnotesize
\begin{eqnarray}\label{eq:utility_quadratic}
&&(1-\lambda) \left(  \frac{n_B^{2\alpha}}{(n^\alpha_A + n^\alpha_B)^2}  \cL^B + \frac{n_A^{2\alpha}}{(n^\alpha_A + n^\alpha_B)^2}   \cL^A  +  \frac{n^{2\alpha}_A \sigma^2_A + n^{2\alpha}_B \sigma^2_B}{(n^\alpha_A + n^\alpha_B)^2}\right) \nonumber \\
 &&  + \lambda \left( \frac{2 n_A n_B}{(n_A + n_B)^{1+\beta}} \left( \cL^A + \cL^B \right) + \frac{ 2 n_A n_B + 2n_A (n_A-1)}{(n_A + n_B)^{1+\beta}} \sigma^2_A + \frac{ 2 n_A n_B + 2n_B (n_B-1)}{(n_A + n_B)^{1+\beta}} \sigma^2_B \right) 
\end{eqnarray}
}

Next, we address the following question: for a given $0<\lambda<1$, how and to what extent does a team with initial composition $(n_A, 0)$ grow? And how does the resulting composition compare with the accuracy-optimal team? We observe that for any strictly positive value of $\lambda$,  the team fails to add the appropriate number of type $B$ members, leading to accuracy inefficiencies. For ease of exposition, throughout this section we assume $\sigma_A = \sigma_B = 0$.

\hhedit{\xhdr{Outline of the analysis} Our theoretical analysis focuses on deriving \emph{closed-form solutions} for the edge cases of $\alpha=1$ and $\beta \in \{0,1\}$. These particular settings are natural to study, because $\alpha=1$ corresponds to a meaningful, common aggregation mechanism \hhedit{(i.e., simple averaging, which is often utilized in practice and has been advocated as a good rule of thumb (Makridakis-Winkler’1983, Clemen-Winkler’1986, Clemen’1989, Armstrong’2001).} $\beta \in \{0,1\}$ capture whether perceptions of conflict within the team depend on the \emph{relative} or the \emph{absolute} size of the types. The analysis of these extremes offers several non-trivial observations, as will be stated shortly. For other values of $\alpha$ and $\beta$, we provide simulation results (Figure~\ref{fig:plain_lambda_mid_alpha_high}) showing that the effects observed at the edge cases continue to hold, but they interact with each other in potentially interesting ways. }

\begin{theorem}[Utility-optimal composition for $\alpha=1, \beta=0$]\label{prop:utility_optimal_beta_0}
Consider a team with an initial composition of $n_A>0$ members of type $A$ and no member of type $B$. Fix $\lambda$ for the team. The optimal number of type $B$ members whose addition maximizes the team's utility is equal to 
$n^*_B = \frac{\lambda(\cL^A + \cL^B)n_A^2 - (1-\lambda)\cL^A n_A}{-\lambda(\cL^A + \cL^B )n^A  -  (1-\lambda)\cL^B}$.
\end{theorem}
\begin{proof}
Recall that when $\sigma^2_A = \sigma^2_B = 0$, the disagreement term in the team's objective (\ref{eq:utility_quadratic}) simplifies to:
{\footnotesize
\begin{equation*}
\frac{1}{(n_A + n_B)} \sum_{i,j \in T} d(i,j) = \frac{2 n_A n_B}{(n_A + n_B)} \left( \cL^A + \cL^B \right)
\end{equation*}
}
Taking the derivative of the right hand side with respect to $n_B$, we obtain:
$
2\left( \cL^A + \cL^B \right) \frac{n_A(n_A+n_B) -  n_A n_B}{(n_A + n_B)^2}
=  2\left( \cL^A + \cL^B \right) \frac{n_A^2}{(n_A + n_B)^2}.
$ 
Note that the above is always positive, and decreasing in $n_B$. So if the cost function (the $\lambda$-weighted sum of disagreement and loss) has a zero, it must be before the zero of accuracy derivative, that is, before $n_A   \frac{\cL^A}{\cL^B}$. To see where precisely the zero lies, we can write:

{\footnotesize
\begin{eqnarray*}
&&  2 \lambda \left( \cL^A + \cL^B \right) \frac{n_A^2}{(n_A + n_B)^2} + (1-\lambda) \left( \frac{2n_Bn_A}{(n_A + n_B)^3}  \cL^B + \frac{-2n_A^2}{(n_A + n_B)^3}   \cL^A \right)  = 0\\
&\Leftrightarrow &   2 \lambda \left( \cL^A + \cL^B \right) n_A(n_A + n_B) +  (1-\lambda) \left(2n_B  \cL^B -2n_A  \cL^A \right)  = 0\\
&\Leftrightarrow &   n_B = \frac{2\lambda(\cL^A + \cL^B)n_A^2 -2 (1-\lambda)\cL^A n_A}{-2\lambda(\cL^A + \cL^B )n^A  - 2 (1-\lambda)\cL^B}.
\end{eqnarray*}
}

\end{proof}

Figure~\ref{fig:plain_lambda_mid_alpha_1_beta_0} shows the team growth dynamics for $\lambda=0.02$ when $\alpha=1$, $\beta=0$, and under several different regimes of $(\cL^A, \cL^B)$. Note that team formation dynamics can get stuck in local utility optima, failing to achieve accuracy-optimal compositions. Additionally, the team composition is highly sensitive to its initial composition---which can be thought of as a form of path-dependence\hhedit{, as formalized through the corollary below}.

\hhedit{
\begin{corollary}[Convergence of team growth dynamics for $\alpha=1, \beta=0$]
Consider a team with the initial composition of $(n_A,n_B)$. Without loss of generality, we assume $n_B \leq n_A$. Let $n^*_B = \frac{2\lambda(\cL^A + \cL^B)n_A^2 -2 (1-\lambda)\cL^A n_A}{-2\lambda(\cL^A + \cL^B )n^A  + 2 (1-\lambda)\cL^B}$. Regardless of the order in which agents of each type arrive, the team growth dynamics converges to $(n'_A, n'_B)$ where
\begin{itemize}
    \item $n'_A = n_A$ and $n'_B = n^*_B$ if $n_B \leq n^*_B$ (utility-optimal composition).
    \item $n'_A = n_A$ and $n'_B = n_B$ otherwise (sub-optimal composition).
\end{itemize}
\end{corollary}
}

\begin{figure*}[t!]
    \centering
    \begin{subfigure}[b]{0.32\textwidth}
        \includegraphics[width=\textwidth]{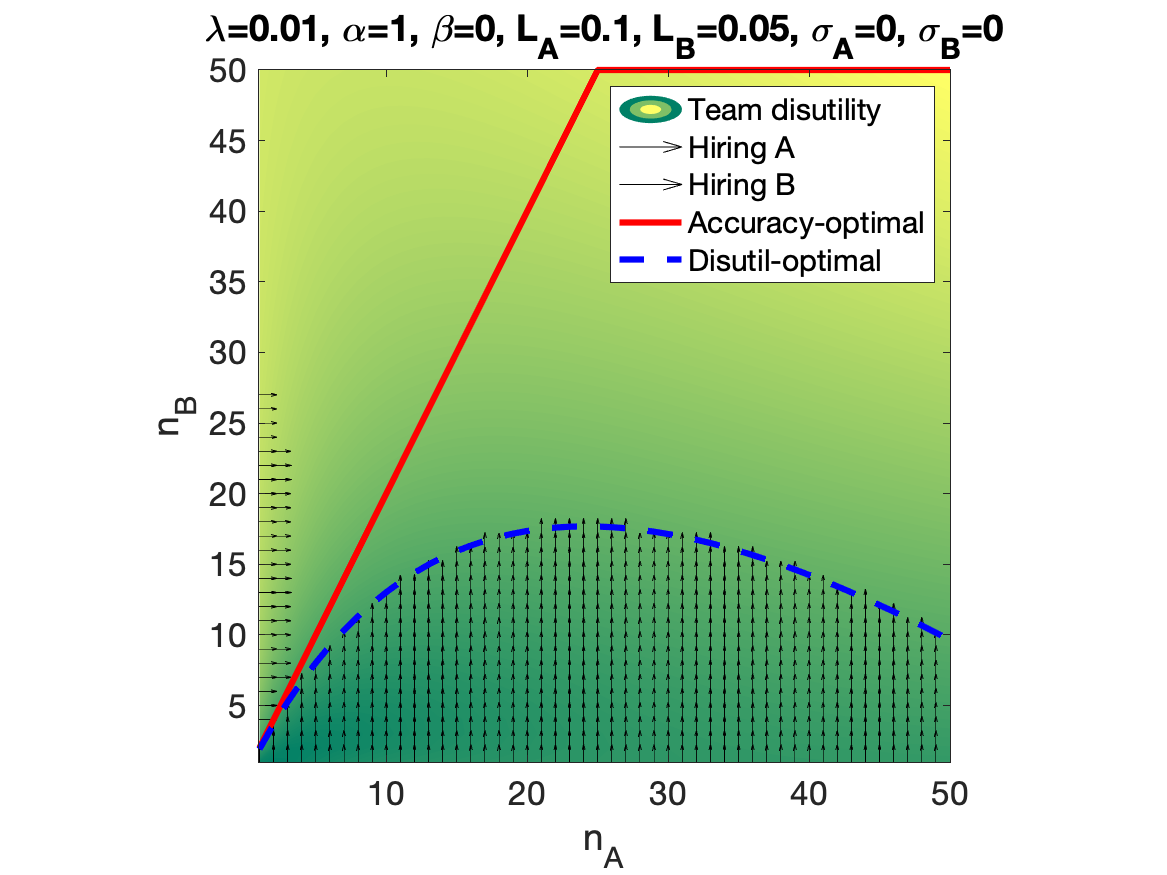}
    \end{subfigure}
    \begin{subfigure}[b]{0.32\textwidth}
        \includegraphics[width=\textwidth]{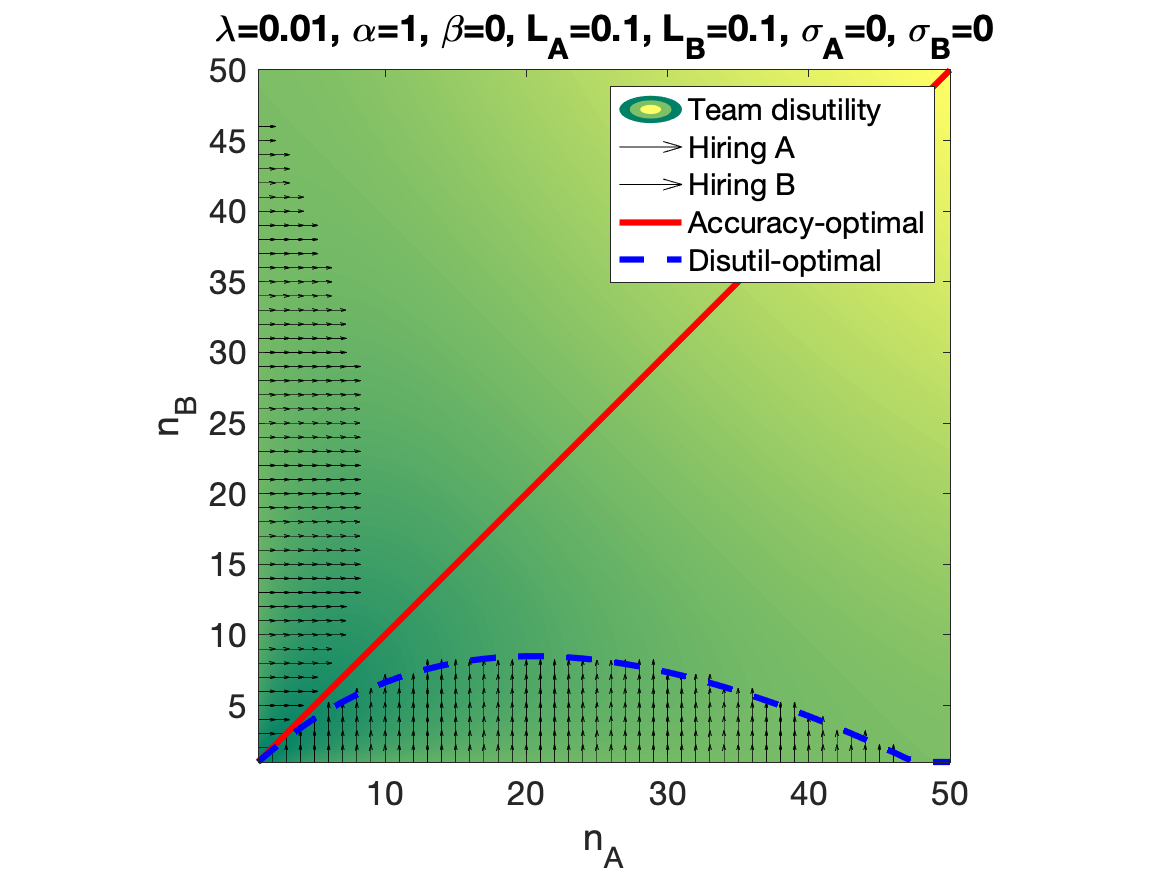}
    \end{subfigure}
    \begin{subfigure}[b]{0.32\textwidth}
        \includegraphics[width=\textwidth]{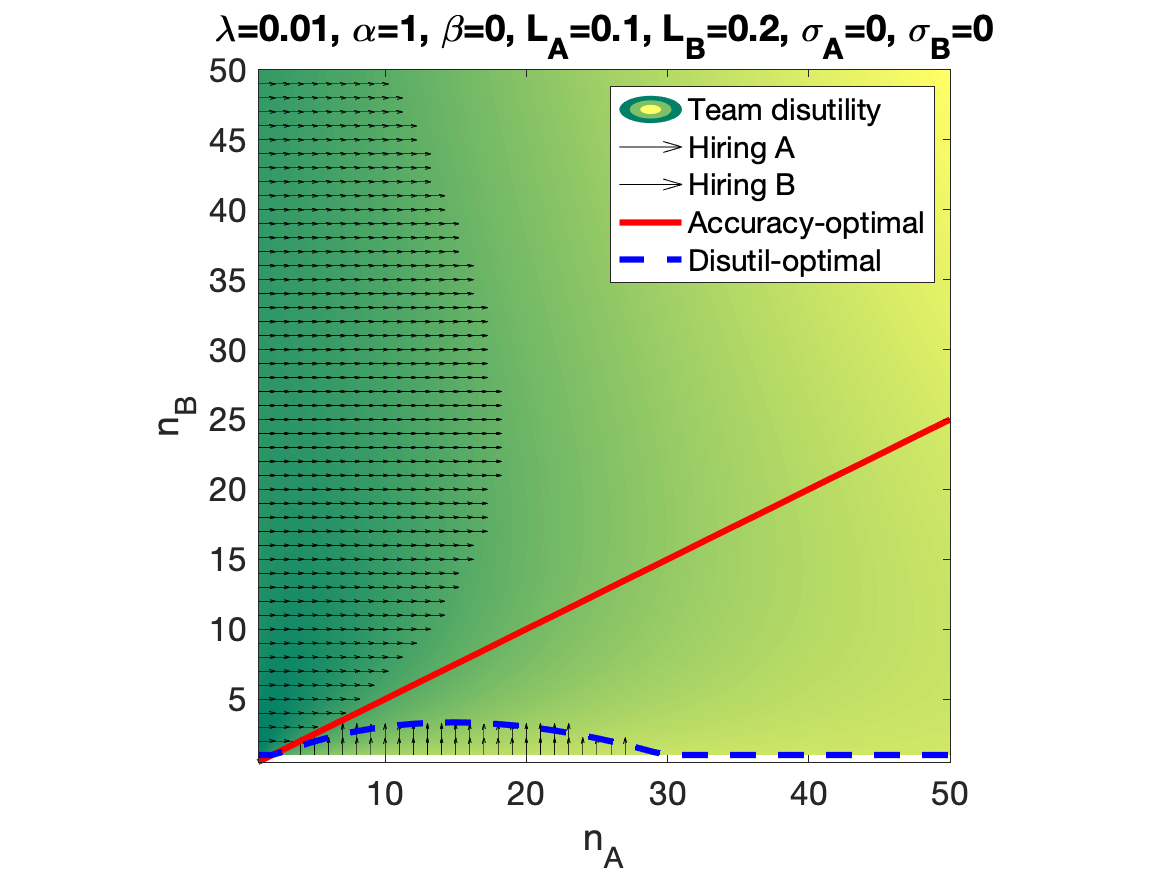}
    \end{subfigure}
    \caption{Visualization of team formation dynamics for an intermediate value of $\lambda$ ($\lambda=0.01$). $\alpha=1$, $\beta=0$, $\cL^A = 0.1$, and $\cL^B = \{0.05, 0.1, 0.2\}$. Team growth dynamics can get stuck in local optima, failing to achieve accuracy-optimal compositions.}\label{fig:plain_lambda_mid_alpha_1_beta_0}
    \vspace{-3mm}
\end{figure*}

\begin{theorem}[Utility-optimal composition for $\alpha=1, \beta=1$]\label{prop:utility_optimal_beta_1}
Consider a team with an initial composition of $n_A>0$ members of type $A$ and no member of type $B$. Fix $\lambda$ for the team. The optimal number of type $B$ members whose addition maximizes the team's utility is equal to $n^*_B  = n_A \frac{\lambda(\cL^A + \cL^B) -(1-\lambda) \cL^A}{\lambda(\cL^A + \cL^B )  -(1-\lambda)\cL^B}$.
\end{theorem}
\begin{proof}
Recall that when $\sigma^2_A = \sigma^2_B = 0$, the disagreement term in the team's objective (\ref{eq:utility_quadratic}) simplifies to $\frac{2n_A n_B}{(n_A + n_B)^2} \left( \cL^A + \cL^B \right)$. 
Taking the derivative of this function with respect to $n_B$, we obtain:
{\footnotesize
\begin{equation*}
2\left( \cL^A + \cL^B \right) \frac{n_A(n_A+n_B)^2 - 2 n_A n_B (n_A + n_B)}{(n_A + n_B)^4}
=  2\left( \cL^A + \cL^B \right) \frac{n_A(n_A-n_B)}{(n_A + n_B)^3}
\end{equation*}
}
Note that the above is always positive, and decreasing in $n_B$ as long as as $n_B \leq n_A$. So if the cost function (the $\lambda$-weighted sum of disagreement and loss) has a zero, it must be before $n_A   \frac{\cL^A}{\cL^B}$. To see where precisely the zero lies, we can write:
{\footnotesize
\begin{eqnarray*}
&&  2\lambda \left( \cL^A + \cL^B \right) \frac{n_A(n_A - n_B)}{(n_A + n_B)^3} + (1-\lambda) \left( \frac{2n_Bn_A}{(n_A + n_B)^3}  \cL^B + \frac{-2n_A^2}{(n_A + n_B)^3}   \cL^A \right)  = 0\\
&\Leftrightarrow &   2\lambda \left( \cL^A + \cL^B \right) n_A(n_A - n_B) +  (1-\lambda) \left(2n_Bn_A  \cL^B -2n_A^2  \cL^A \right)  = 0\\
&\Leftrightarrow &   2\lambda\left( \cL^A + \cL^B \right) (n_A - n_B) +  (1-\lambda) \left(2n_B  \cL^B -2n_A \cL^A \right)  = 0\\
&\Leftrightarrow &   n_B = n_A \frac{2\lambda(\cL^A + \cL^B) -2 (1-\lambda) \cL^A}{2\lambda(\cL^A + \cL^B )  - 2 (1-\lambda)\cL^B}.
\end{eqnarray*}
}
\end{proof}

Figure~\ref{fig:plain_lambda_mid_alpha_1_beta_1} illustrates the team growth dynamics for $\lambda=0.02$ when $\alpha=1$, $\beta=1$, and under several different regimes of $(\cL^A, \cL^B)$. Team formation dynamics continue to converge to utility optima, failing
to achieve accuracy-optimal compositions. Note, however, the different patterns of inefficiency in the case of $\beta=0$ and $\beta=1$.

\hhedit{
\begin{corollary}[Convergence of team growth dynamics for $\alpha=1, \beta=1$]
Consider a team with the initial composition of $(n_A,n_B)$. Let $n^*_B = n_A \frac{\lambda(\cL^A + \cL^B) -(1-\lambda) \cL^A}{\lambda(\cL^A + \cL^B )  - (1-\lambda)\cL^B}$. Regardless of the order in which agents of each type arrive, the team growth dynamics converges to $(n'_A, n'_B)$ where
\begin{itemize}
    \item $n'_A = n_A$ and $n'_B = n^*_B$ if $n_B \leq n^*_B$.
    \item $n'_A = n_B \frac{\lambda(\cL^B + \cL^A) -(1-\lambda) \cL^B}{\lambda(\cL^A + \cL^B )  - (1-\lambda)\cL^A}$ and $n'_B = n_B$ otherwise.
\end{itemize}
\end{corollary}
}

\begin{figure*}[t!]
    \centering
    \begin{subfigure}[b]{0.32\textwidth}
        \includegraphics[width=\textwidth]{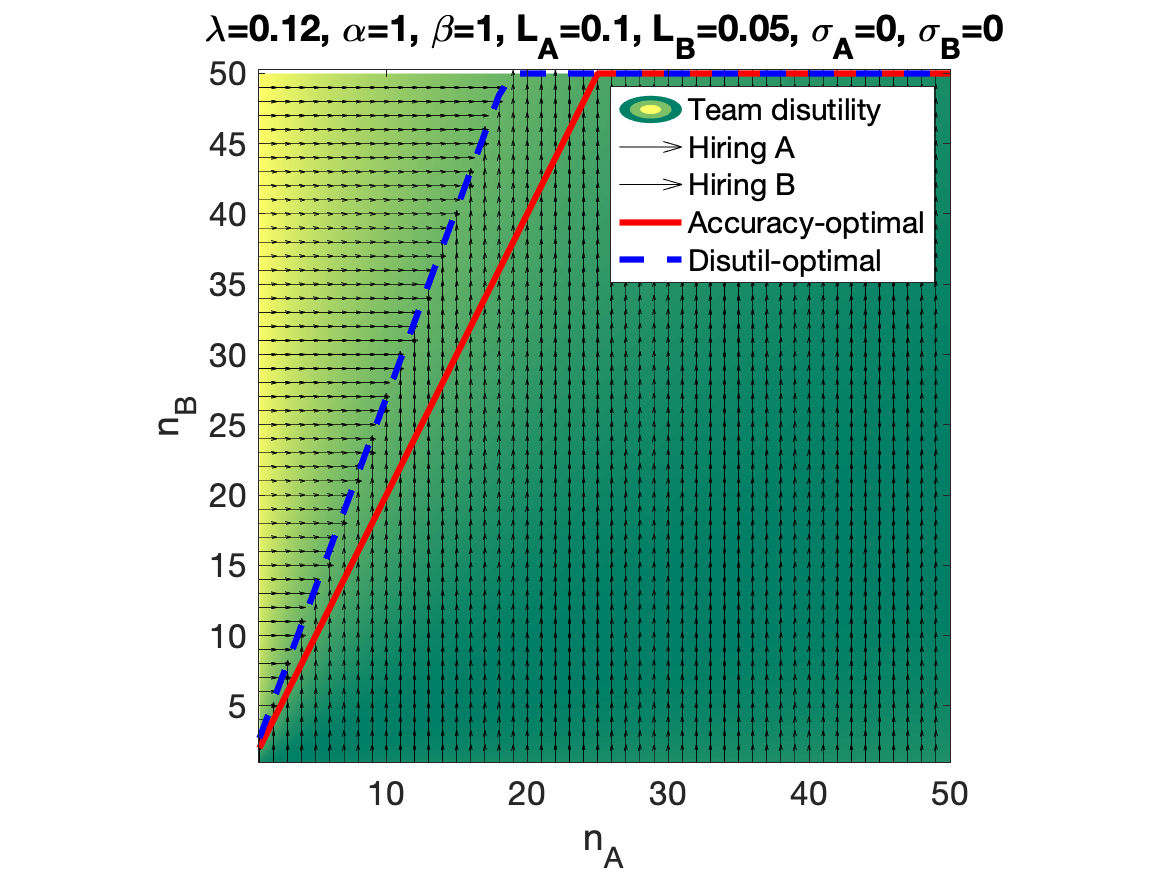}
    \end{subfigure}
    \begin{subfigure}[b]{0.32\textwidth}
        \includegraphics[width=\textwidth]{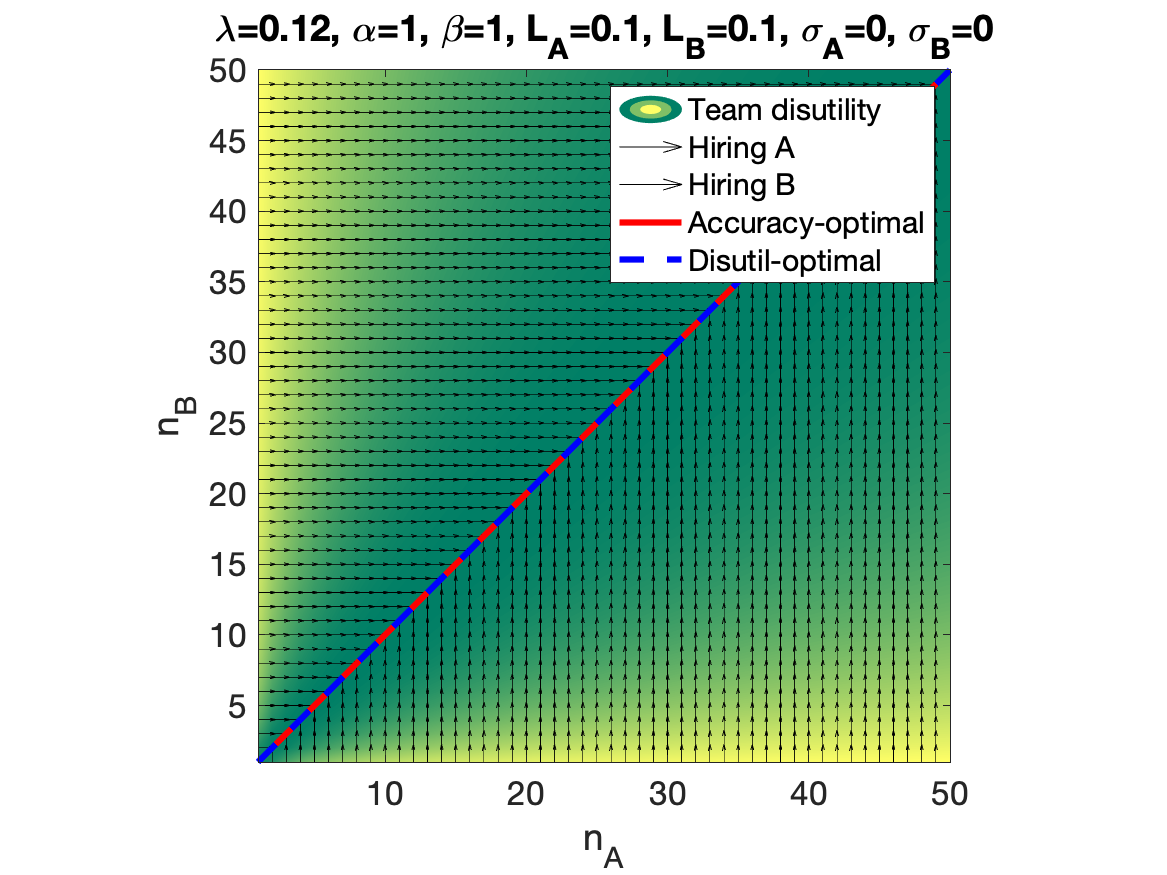}
    \end{subfigure}
    \begin{subfigure}[b]{0.32\textwidth}
        \includegraphics[width=\textwidth]{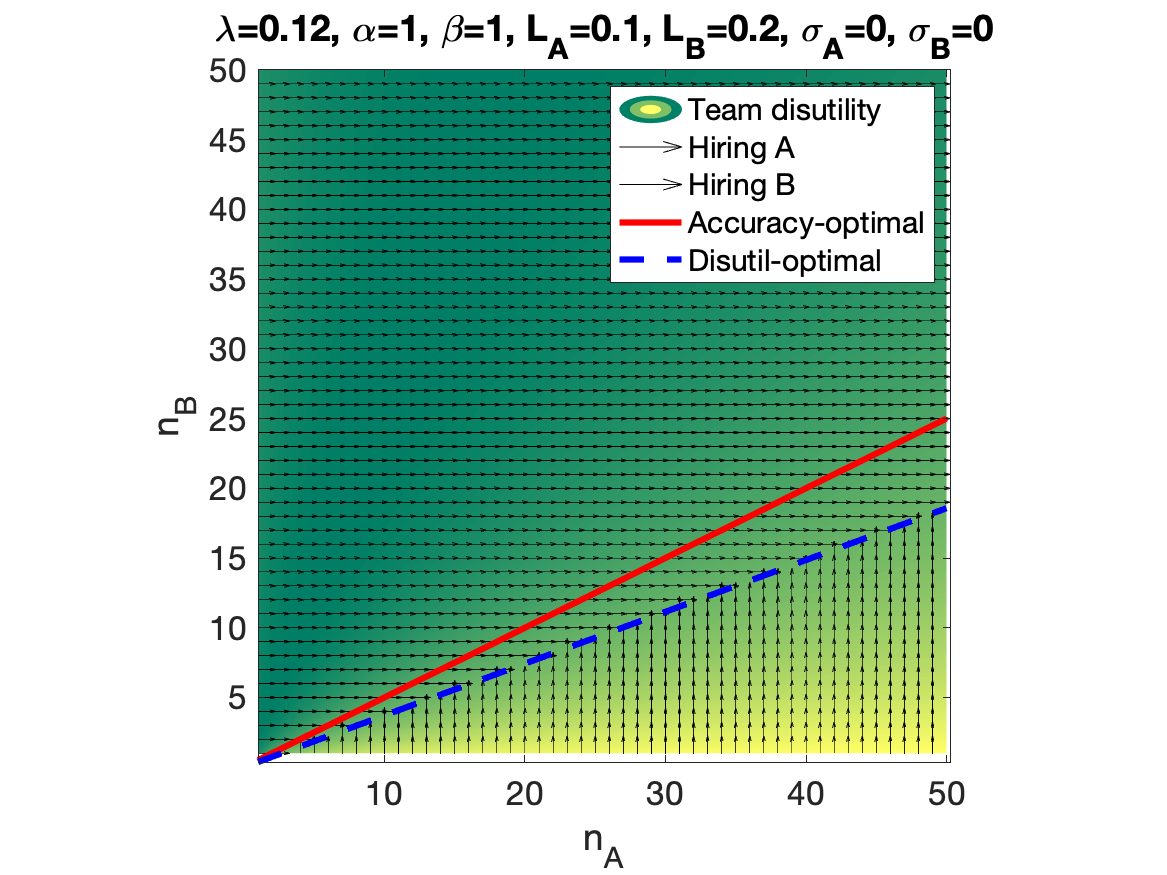}
    \end{subfigure}
    \caption{Visualization of team formation dynamics for an intermediate value of $\lambda$ ($\lambda=0.12$). $\alpha=1$, $\beta=1$, $\cL^A = 0.1$, $\cL^B = \{0.05, 0.1, 0.2\}$ and $\var(\epsilon) = 0$. Team formation dynamics converge to utility optima, failing to achieve accuracy-optimal compositions.}\label{fig:plain_lambda_mid_alpha_1_beta_1}
    \vspace{-4mm}
\end{figure*}

Figure~\ref{fig:plain_lambda_mid_alpha_high} illustrates the team growth dynamics for $\lambda=0.05$ when $\alpha$ is high. In general, we observe that high $\alpha$ induces a lower bound on the number of less-represented type members needed to make increasing the type's representation in the team beneficial. Additionally, larger $\beta$ values encourage a dominant majority to bring on more members of its own to reduce disagreement---even if that comes at the cost of accuracy.

\subsection{Takeaways from the Analysis}
\xhdr{The path-dependent nature of inefficiencies} Through the analysis in this Section, we observe that the initial composition of the team plays an important role in its eventual composition. As an illustrative example of the different effects at work, consider Figure~\ref{fig:plain_lambda_mid_alpha_high} (b). The initial composition of the team dictates whether (a) the team remains at its initial makeup, (b) it adds members to the less-represented type to move toward greater accuracy, or (c) it continues adding to the more-represented type in a way that overpowers the less-represented type.  While the exact dynamics are specific to our model, this general family of observations has important implications for teams in organizations more generally: that the initial composition can have a significant effect on the direction in which the team grows. 

\xhdr{The role of the aggregation mechanism}
It is also interesting to note the ways in which varying the aggregation parameter $\alpha$ has an effect on the team growth dynamics and incentives for the team to add members of each group.  This suggests more generally some of the mechanisms whereby aggregation can influence decisions about group composition.  There are interesting analogies to other contexts that exhibit a link between aggregation mechanisms and the dynamics of new membership.  For example, while legislative bodies are distinct from problem-solving teams, there is an interesting analogy to issues such as the way in which the prospect of statehood for entities like the District of Columbia and Puerto Rico play out differently in the U.S. House of Representatives, where aggregation is done proportionally to population, and the U.S. Senate, where aggregation is done uniformly across states.
This is precisely a case of the difference in aggregation mechanism implying differences in the politics of new membership (in this case via statehood).

\begin{figure*}[t!]
    \centering
    \begin{subfigure}[b]{0.32\textwidth}
        \includegraphics[width=\textwidth]{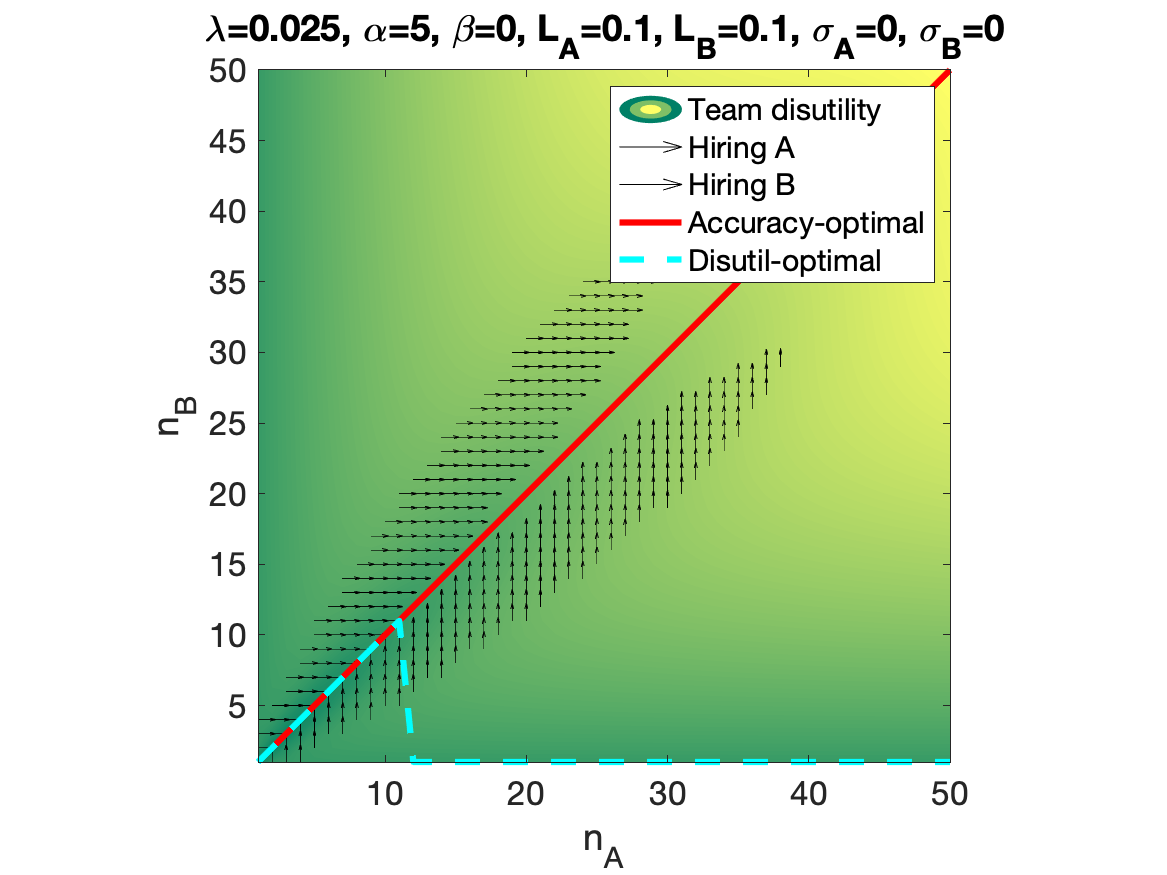}
    \end{subfigure}
    \begin{subfigure}[b]{0.32\textwidth}
        \includegraphics[width=\textwidth]{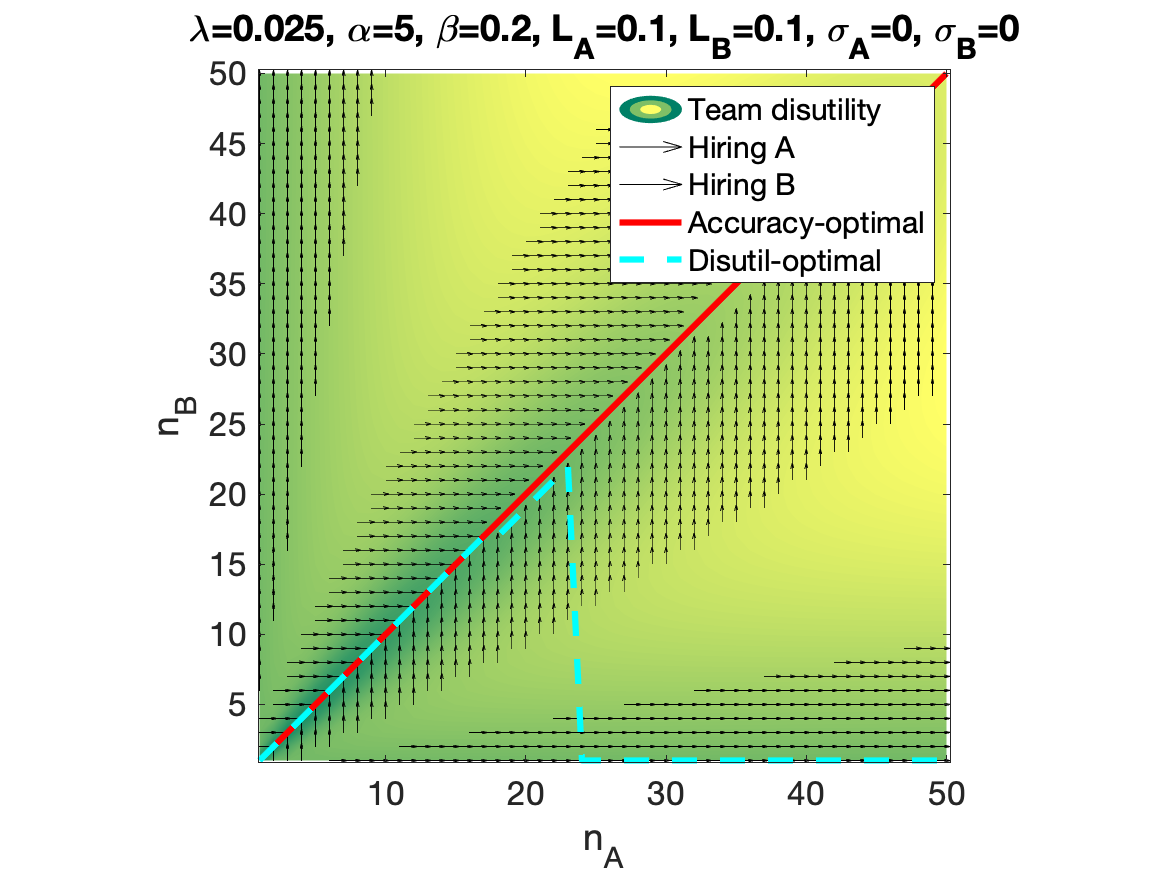}
    \end{subfigure}
    \begin{subfigure}[b]{0.32\textwidth}
        \includegraphics[width=\textwidth]{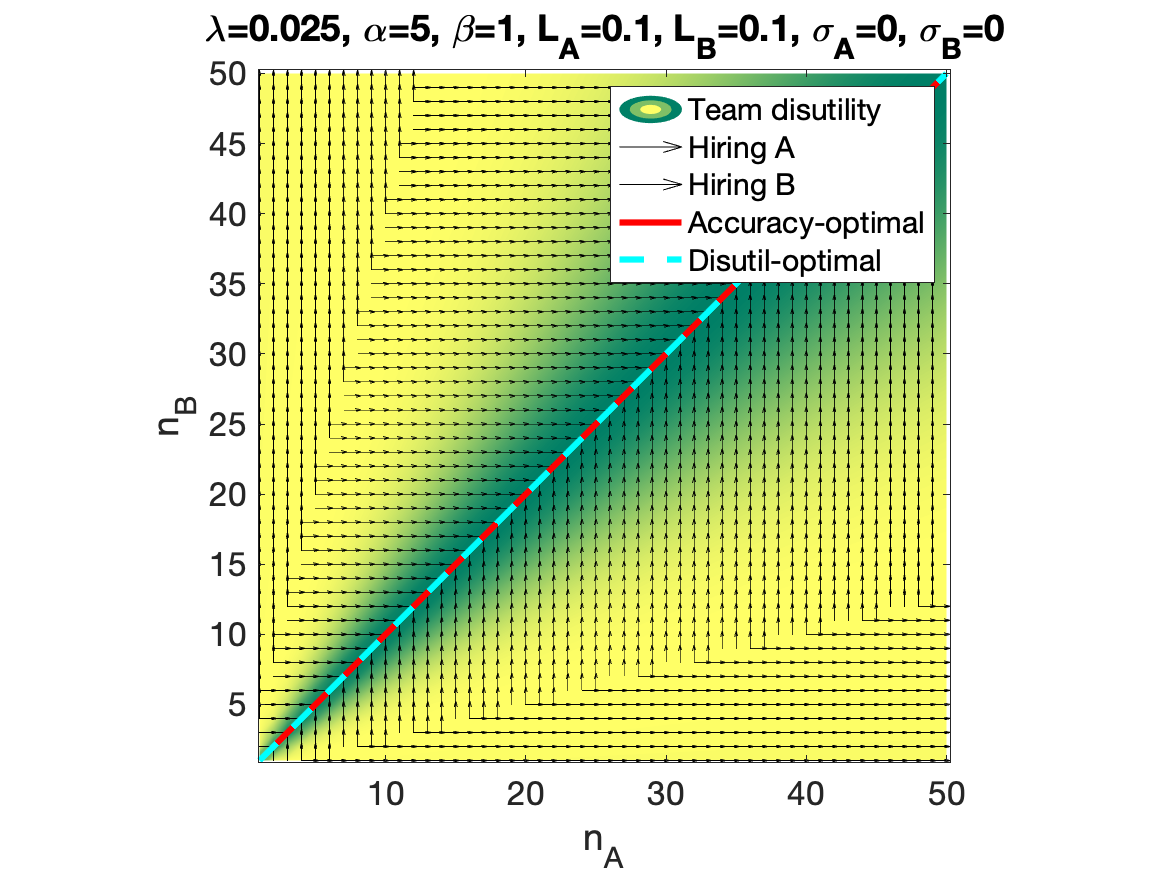}
    \end{subfigure}
    \caption{Visualization of team formation dynamics for an intermediate value of $\lambda$ ($\lambda=0.025$), $\alpha=5$, $\beta=\{0, 0.2, 1\}$, $\cL^A = 0.1$, $\cL^B = 0.1$ and $\sigma_A = \sigma_B = 0$. When $\alpha$ is high, adding a member of the less-represented type is only beneficial if it has a non-negligible impact on accuracy, hence the lower bound on the number of the type's members for their addition to start. (a) since $\beta=0$, there is an upper bound on the number of less-represented group members. (b,c) for larger $\beta$ values, if the majority is sufficiently dominant, hiring more majority members reduces the disagreement, which is beneficial (even if it degrades the accuracy). }\label{fig:plain_lambda_mid_alpha_high}
    \vspace{-3mm}
\end{figure*}

\section{Extensions}\label{sec:extension}

\xhdr{Alternative notions of distance and accuracy loss}
Throughout our analysis, we assumed that the distance and the loss functions take on simple quadratic forms. It is easy to show that our main result (that team composition initially trends toward improving accuracy but stops short of achieving the optimal performance) holds for more generic functional forms.
Consider a distance metric $\delta(.,.)$ capturing disagreements between team members, and a loss function $\ell(.,.)$ capturing the team's predictive loss. For simplicity, let's assume both $\ell$ and $d$ are continuous and differentiable.
Let's define the following pieces of notation for convenience:
{
$$\tilde{\delta}(n_A,n_B) := \lambda \times \frac{1}{(n_A + n_B)^{\beta}}  \Exp_{\vx \sim P}\left[ \delta(f_A(\vx), f_B(\vx)) \right] \text{\indent and \indent}\tilde{\ell}(n_A,n_B) := \Exp_{\vx \sim P}\left[ \ell(\cG_{n_A,n_B}(\vx), y) \right].$$
}
With similar reasoning as that presented in the proof of Proposition~\ref{prop:utility_optimal_beta_0}, we can show the following:

\begin{proposition}[informal statement]
Consider a team with an initial composition
of $n_A > 0$ members of type A and no member of type B. Fix $\lambda$ for the team. 
Suppose $\tilde{\delta}(.,.)$ and $\tilde{\ell}(.,.)$ are both differentiable, and the following conditions hold:
\begin{enumerate}
\item $\tilde{\delta}(n_A,.)$ is \emph{concave} and \emph{increasing} in the number of the less-represented group members, $n_B$. 
\item $\tilde{\ell}(n_A,.)$, is initially \emph{decreasing} and \emph{convex},  but becomes and remains \emph{increasing} thereafter.
\end{enumerate}
Then the optimal number of type B
members whose addition maximizes the team’s utility is strictly less than $n_A$.
\end{proposition}
\begin{proof}[Proof sketch]
Note that the first condition implies $\frac{\partial}{\partial n_B}\tilde{\delta}(n_A,.)$ is positive and decreasing (with $c \geq 0$ as its potential asymptote).  Additionally,  the second condition implies that the $\frac{\partial}{\partial n_B}\tilde{\ell}(n_A,.)$ is initially negative, but reaches zero at some point and remains positive thereafter.  The derivative of the sum is equal to the sum of derivatives, so the derivative of the team's objection function with respect to $n_B$ is equal to 
$ \lambda \times \frac{\partial}{\partial n_B}\tilde{\delta}(n_A,.) + (1-\lambda) \times \frac{\partial}{\partial n_B}\tilde{\ell}(n_A,.).$
Therefore, if the above derivative has a zero, it must lie before the zero of the accuracy loss term. 
\end{proof}
\hhedit{We remark that with the appropriate choice of $\tilde{\delta}$ and $\tilde{\ell}$, the utility-maximizing number of type $B$ members can be arbitrarily close to the number needed for optimizing accuracy. }

\xhdr{More than two predictive types}
In the analysis in Section~\ref{sec:analysis}, we assumed that agents belong to one of the two types: $A$ or $B$. As we show in Appendix~\ref{sec:three_types}, our derivations readily generalize to three types ($A$, $B$, and $C$), where $f^*(\vx) = \theta^A(\vx) + \theta^B(\vx) + \theta^C(\vx)$. Figure~\ref{fig:three_types} below illustrates the team growth dynamics for three equally accurate predictive types ($\cL^A = \cL^B =\cL^C =0.1$) where $\lambda= 0.025$, $\alpha=5$, and $\beta=0.1$.

\begin{figure}[h!]
    \centering
    \vspace{-4mm}
     \includegraphics[width=0.6\textwidth]{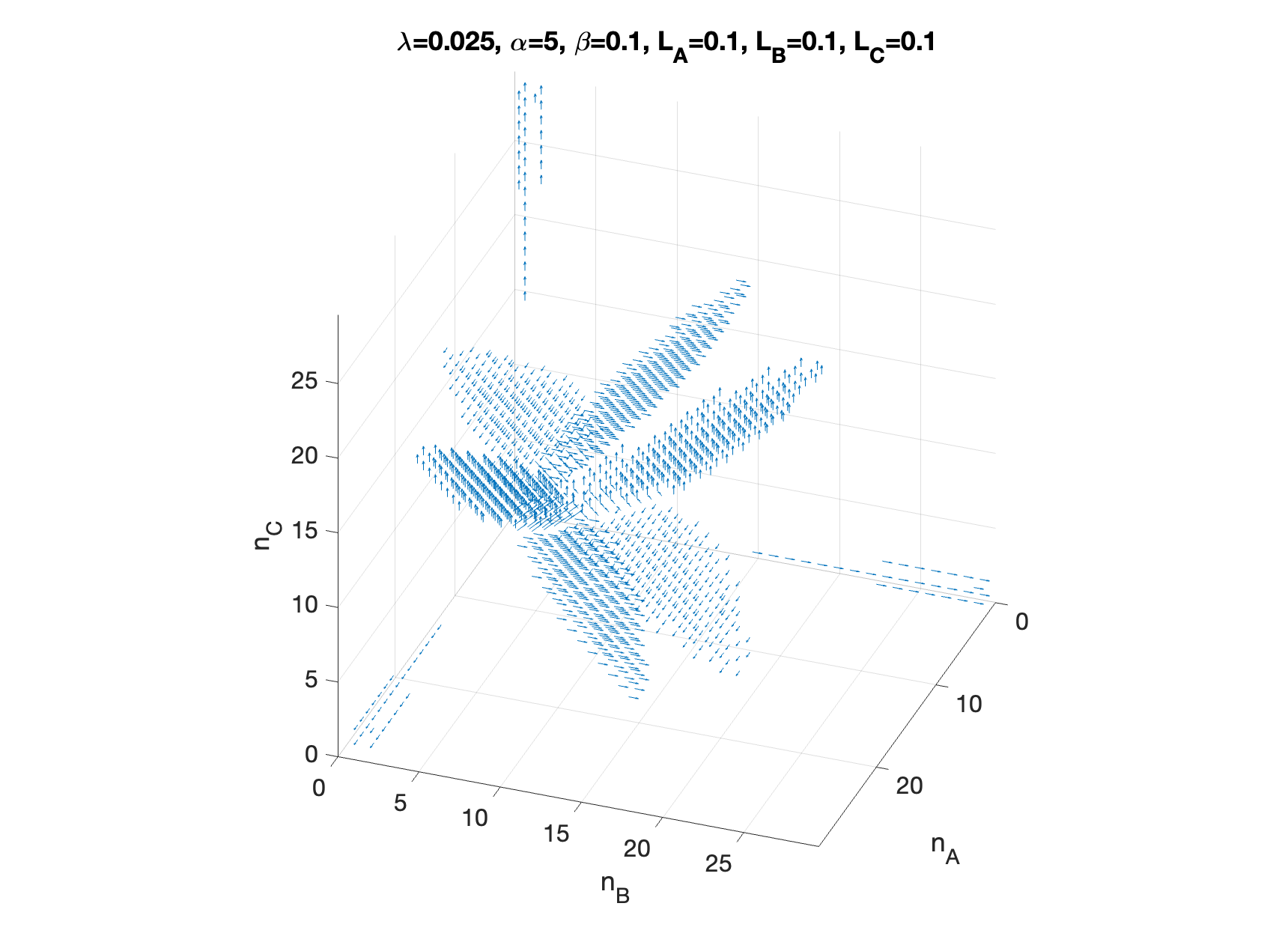}
    \caption{Visualization of team formation dynamics for three types. Trends are similar to that of two types.
    }\label{fig:three_types}
    \vspace{-4mm}
\end{figure}

\xhdr{The role of biased accuracy-gain assessments}
We assumed throughout that teams could perfectly (i.e., without bias and noise) estimate the accuracy gains of adding a new member of each type. While this is a common assumption in prior work, in reality, such assessment may be biased (e.g., optimistic or pessimistic) and noisy. Here, let's consider the biased case, as demonstrated via the example in Figure~\ref{fig:biased_assessment}. (Figure~\ref{fig:noisy_biased_assessment} in the Appendix illustrates the case where assessments are both biased and noisy). In the absence of bias (Figure~\ref{fig:biased_assessment}, (b)), we observed that once teams reach an accuracy-optimal composition, they cease to grow any further. Additionally, adding a new member of type A and a new member of type B could not simultaneously improve the team's utility. When accuracy gain assessments are over-estimated (Figure~\ref{fig:biased_assessment}, (c)), however, the same teams may continue to grow beyond accuracy optimal compositions, and they may find themselves in situations where adding a new member of \emph{any} type is beneficial. As an example, compare the dynamics at $(40,40)$. Conversely, when accuracy gain assessments are under-estimated (Figure~\ref{fig:biased_assessment}, (a)), teams dynamics get stuck in accuray sub-optimal compositions.

\begin{figure*}[t!]
    \centering
    \begin{subfigure}[b]{0.32\textwidth}
        \includegraphics[width=\textwidth]{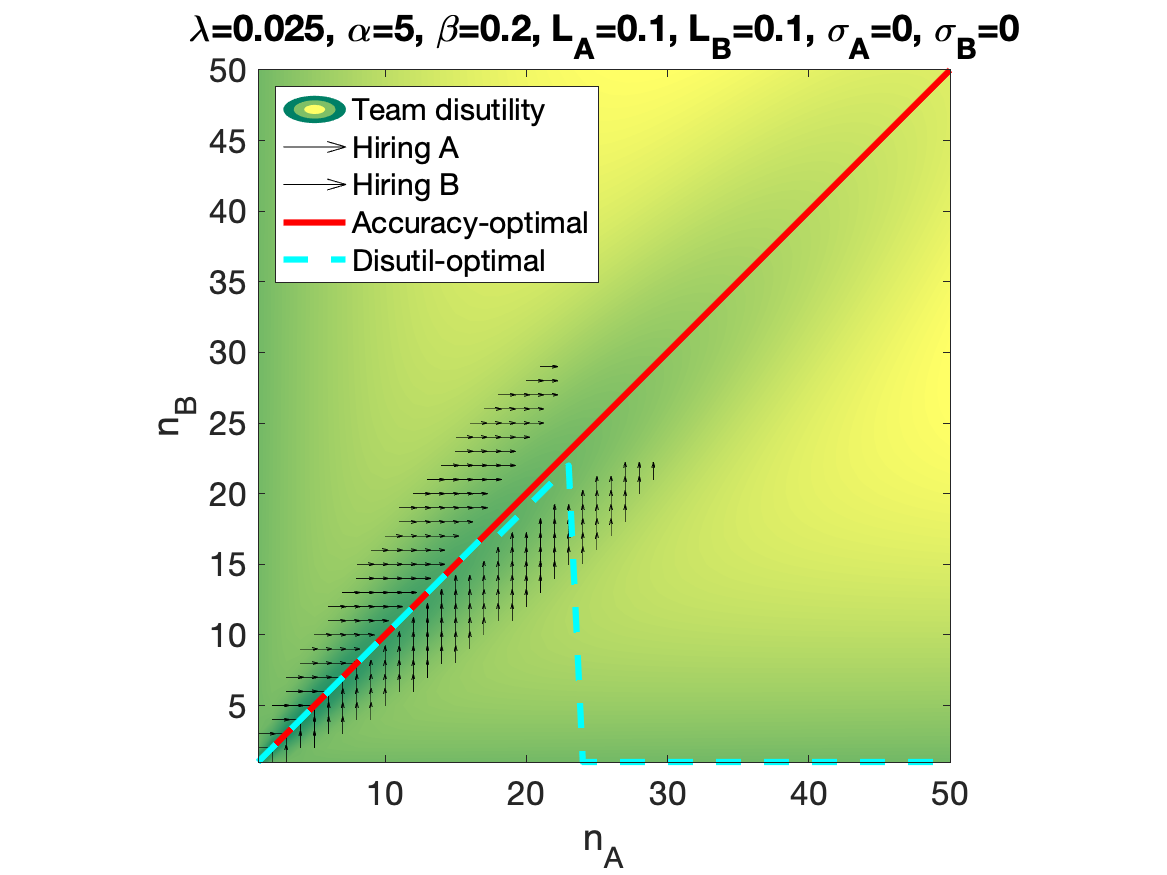}
    \end{subfigure}
    \begin{subfigure}[b]{0.32\textwidth}
        \includegraphics[width=\textwidth]{Figures/lambda_2.5alpha_50_beta_2_LA_10_LB_10_varA_0_varB_0}
    \end{subfigure}
    \begin{subfigure}[b]{0.32\textwidth}
        \includegraphics[width=\textwidth]{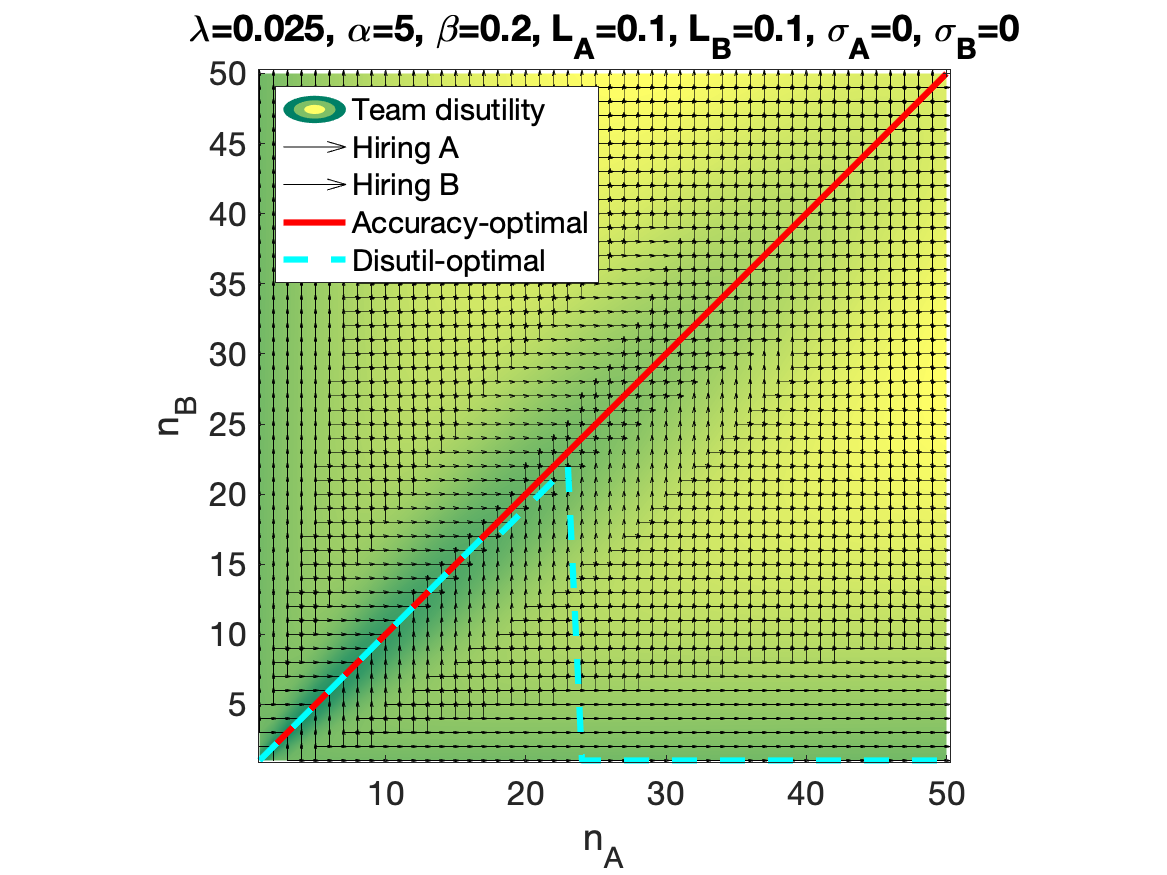}
    \end{subfigure}
    \caption{Visualization of team growth dynamics when assessments of utility gains are biased---as captured by an additive bias term equal to: (a) $0.12$ (under-estimation of gains); (b) $0$ (unbiased estimate); (c) $-0.12$ (over-estimation). (a) The team may not reach accuracy optimality, or (c) it may grow beyond it.}\label{fig:biased_assessment}
    \vspace{-5mm}
\end{figure*}

\vspace{1mm}
\section{Conclusion} \label{sec:conclusion}
This work offered a stylized model of team growth dynamics in the presence of a tension
between informational diversity and affinity bias. Our analysis provides several key observations about the effect of affinity bias on team composition inefficiencies (even an arbitrarily small positive weight on affinity bias leads to inefficiency) and the moderating role of the aggregation mechanism and team size. 
It also shows how the growth dynamics of a team can lead toward optimality for some starting compositions and away from it for others.
\hhedit{Our findings present several actionable insights to improve team growth dynamics. In particular, it shows that awareness of the positive impact of diversity on a team's performance alone will not incentivize high-performing teams to form. But the social planner can positively influence team growth dynamics by adjusting the initial team composition or the aggregation mechanism used to resolve conflicts of opinion.}

We conclude with a discussion of limitations and outline of important directions for future work.

\hhdelete{
\xhdr{Analogy to the Dunning-Kruger effect} The effect of affinity bias in our team formation dynamics creates an indirect yet intriguing analogy to the Dunning-Kruger effect in psychology~\citep{schlosser2013unaware}: the tendency of people with low ability at a given task to overestimate their performance, and the tendency of high performing individuals to underestimate their ability and skills. In other words, there is a tension between performance and self-assessment.  In our work, we identify a different kind of tension between performance and self-assessment, operating on teams rather than individuals: As we saw through the analysis, more diverse teams in our model perform better (in terms of accuracy) but are not as satisfied with their team (in terms of utility), and this dichotomy matches known empirical findings \citep{phillips2009pain,milliken1996searching,lauretta1992effects,watson1993cultural,o1989work}.
}

\vspace{-1mm}
\xhdr{Strategic considerations} Our model considers the incentives of the overall team to improve total utility, but does not account for other kinds of incentives, including \emph{individual} ones. The decision of an individual agent considering whether to join a team or not may be impacted by the proportion of current team members of the same type. For instance, a type $B$ agent may refuse to join if the current number of type $B$ members of the team is below a certain threshold. Agents who already belong to a team may have incentive to exaggerate their opinion in anticipation of their opinions getting aggregated. Under such circumstances, it may be beneficial to utilize
non-uniform/weighted voting schemes both to improve team's accuracy and promote truthfulness. We leave the exploration of such incentives as an important direction for future work.

\vspace{-1mm}
\xhdr{Opinion formation processes} 
Our model does not provide a micro-foundation of opinion dynamics and consensus formation as a function of team members communicating with one another and deliberating. While some of the aggregation functions we study (e.g., uniform average) can be thought of as the outcome of simple opinion formation dynamics, we leave the integration of a more detailed account of opinion evolution in teams for future work.

Last but not least, our analysis relies on a range of additional simplifying abstractions of the team formation process, including (1) the restriction to independent predictive models of the world across the different types of agents; (2) taking accuracy as an appropriate measure of team performance; (3) assuming that teams can accurately estimate the performance gains of increased diversity; and (4) assuming that $\lambda, \alpha,$ and $\beta$ are fixed across types and team compositions. While it would be interesting to see future work that relaxes some of these assumptions, the simplicity of our model enables it to serve as a useful conceptual metaphor capturing an inherent limitation of utility-centric motivations for improved informational diversity: for diverse teams to form and thrive, acknowledging the performance gains of informational diversity alone will not carry the day.

\pagebreak
\bibliographystyle{ACM-Reference-Format}
\bibliography{biblio}
\pagebreak

\appendix
\section{Omitted Technical Material}\label{app:derivations}

\begin{figure*}[t!]
    \centering
    \begin{subfigure}[b]{0.32\textwidth}
        \includegraphics[width=\textwidth]{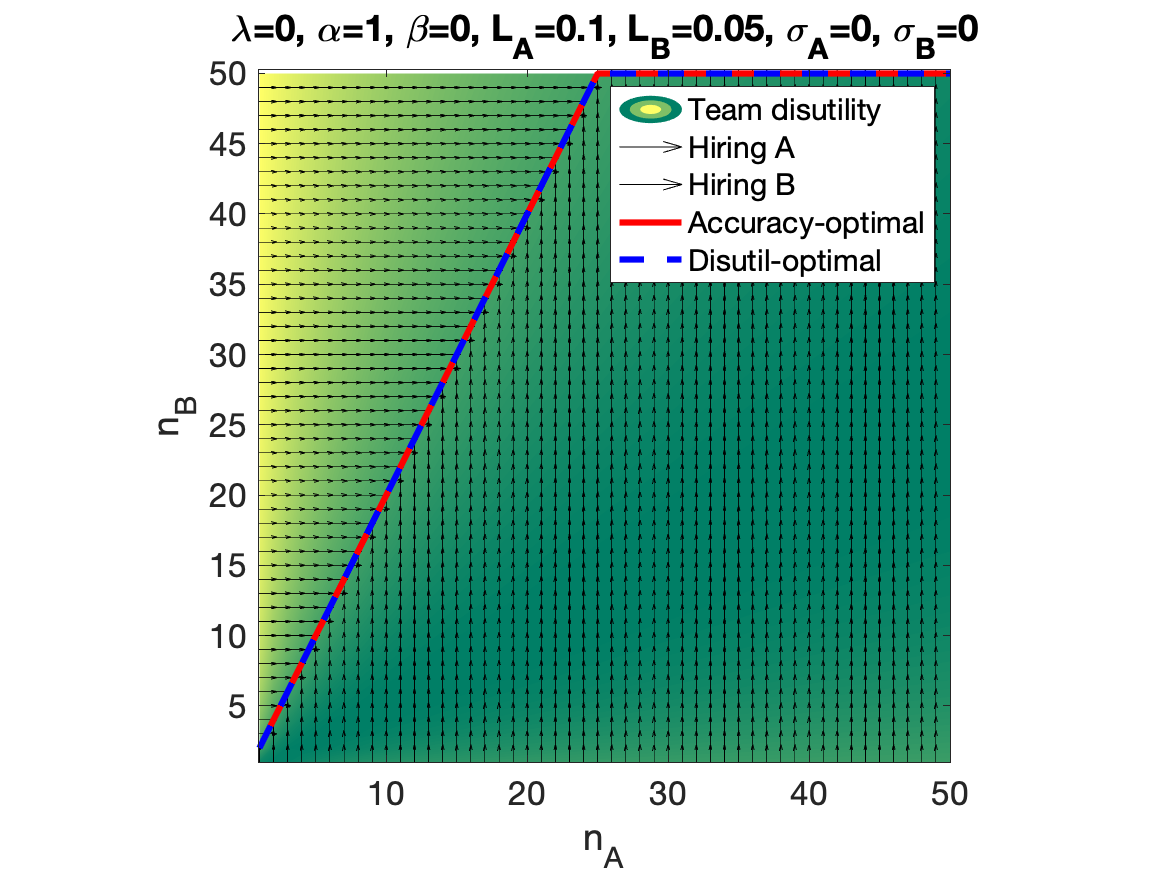}
    \end{subfigure}
    \begin{subfigure}[b]{0.32\textwidth}
        \includegraphics[width=\textwidth]{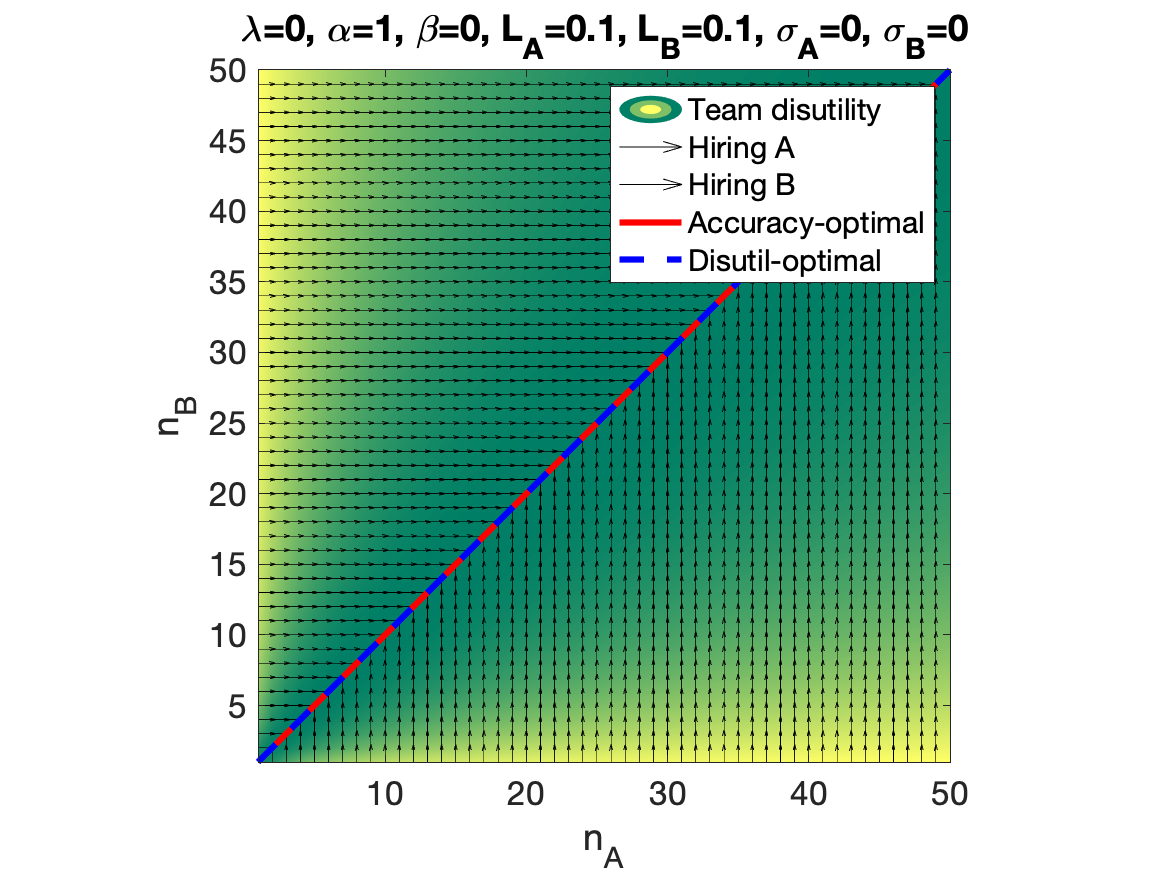}
    \end{subfigure}
    \begin{subfigure}[b]{0.32\textwidth}
        \includegraphics[width=\textwidth]{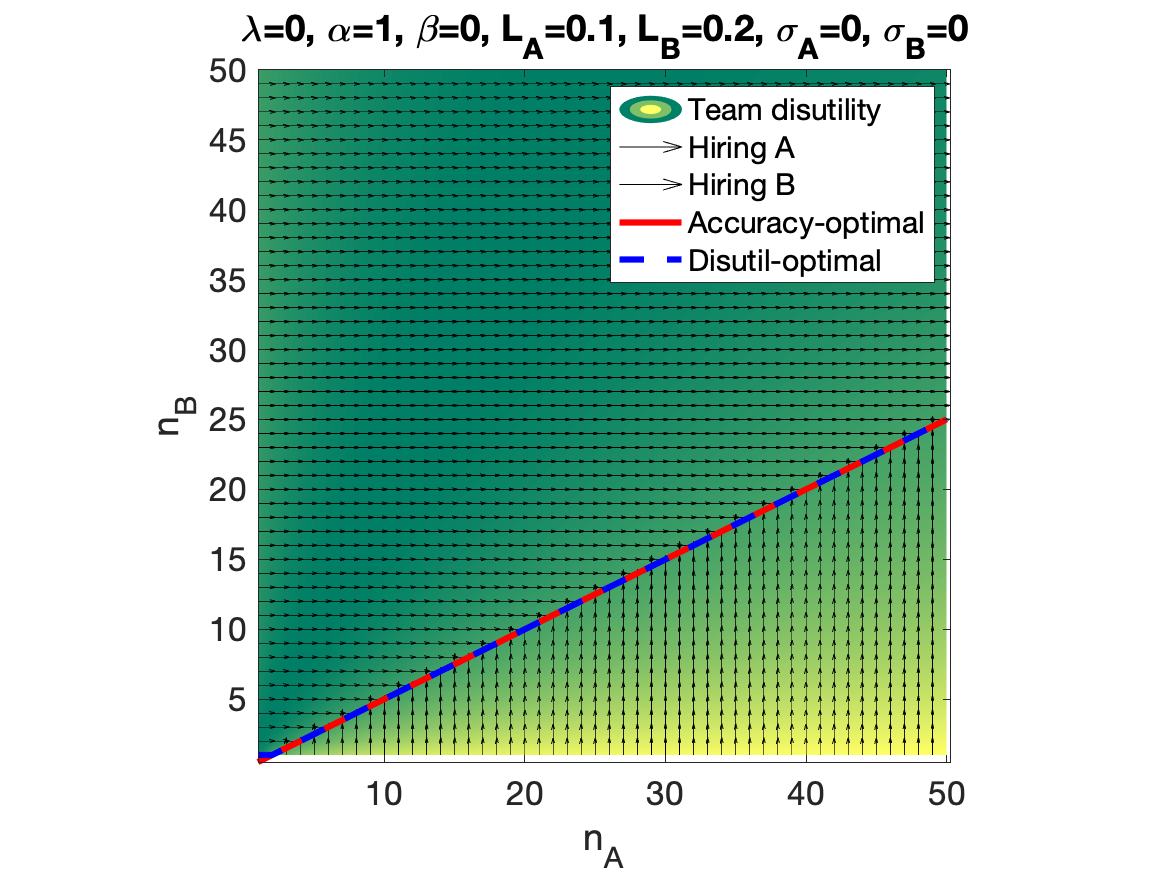}
    \end{subfigure}
    \caption{Visualization of team formation dynamics for $\lambda=0$ when $\alpha=1$ and $\sigma_A=\sigma_B=0$. (Note that in this setting, $n^*_B = \left(\frac{\cL^A}{\cL^B}\right)^{1/\alpha} n_A$.) The team formation dynamics converge to the accuracy/utility-optimal compositions. 
    }\label{fig:plain_lambda_0_alpha_1}
\end{figure*}

\begin{figure*}[t!]
    \centering
    \begin{subfigure}[b]{0.32\textwidth}
        \includegraphics[width=\textwidth]{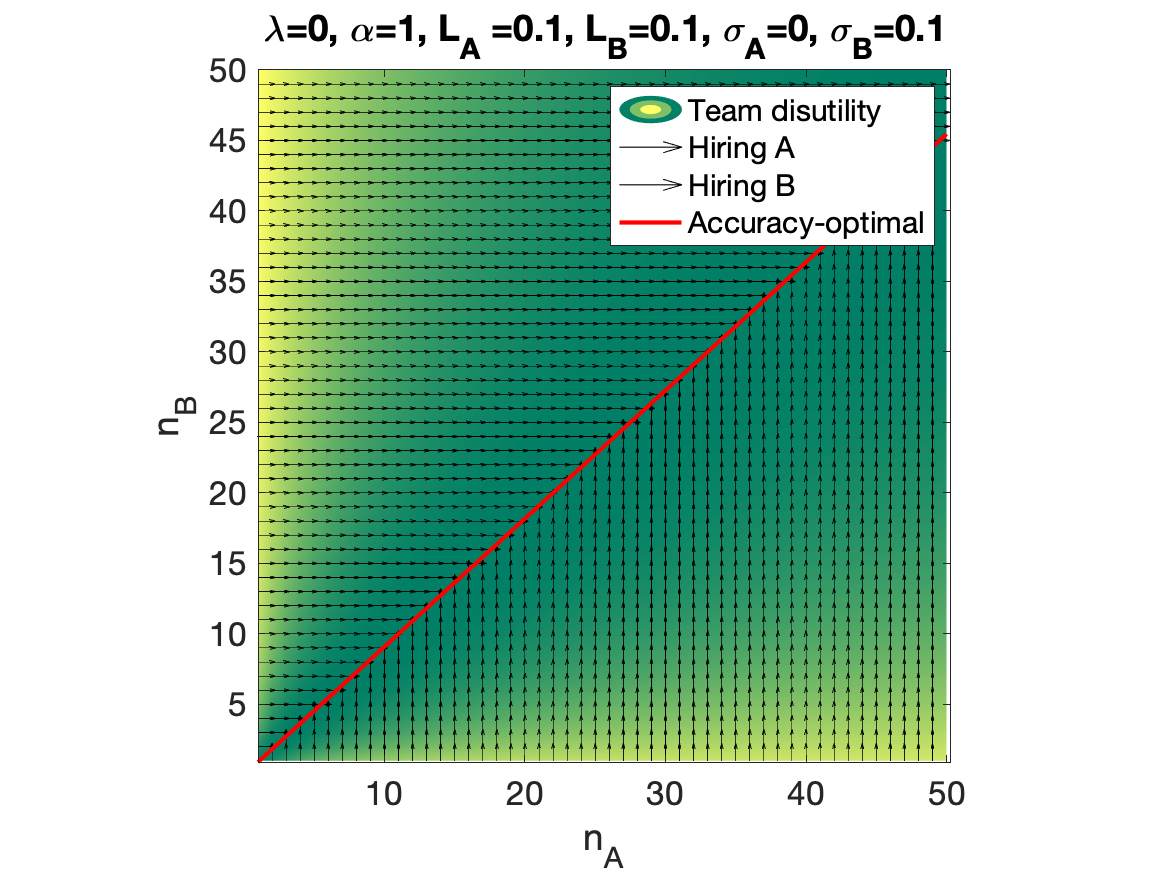}
    \end{subfigure}
    \begin{subfigure}[b]{0.32\textwidth}
        \includegraphics[width=\textwidth]{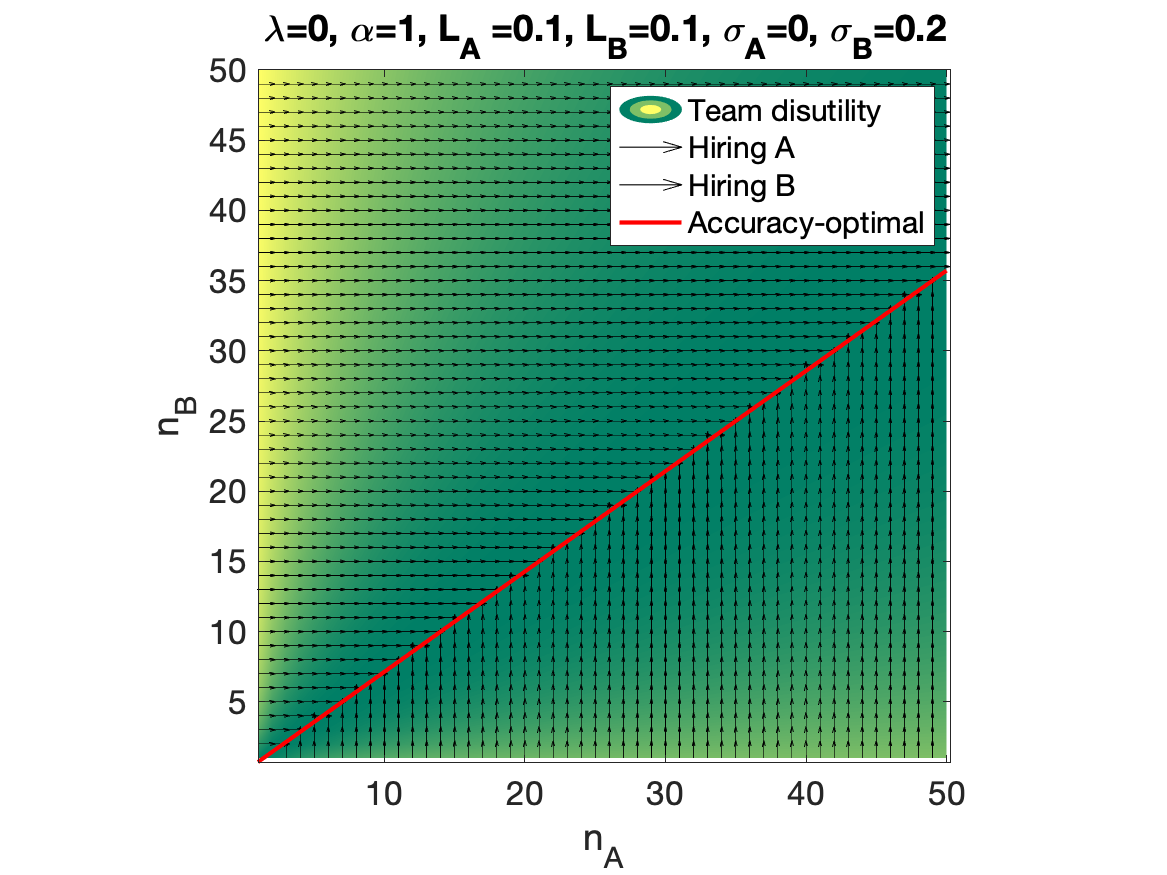}
    \end{subfigure}
    \begin{subfigure}[b]{0.32\textwidth}
        \includegraphics[width=\textwidth]{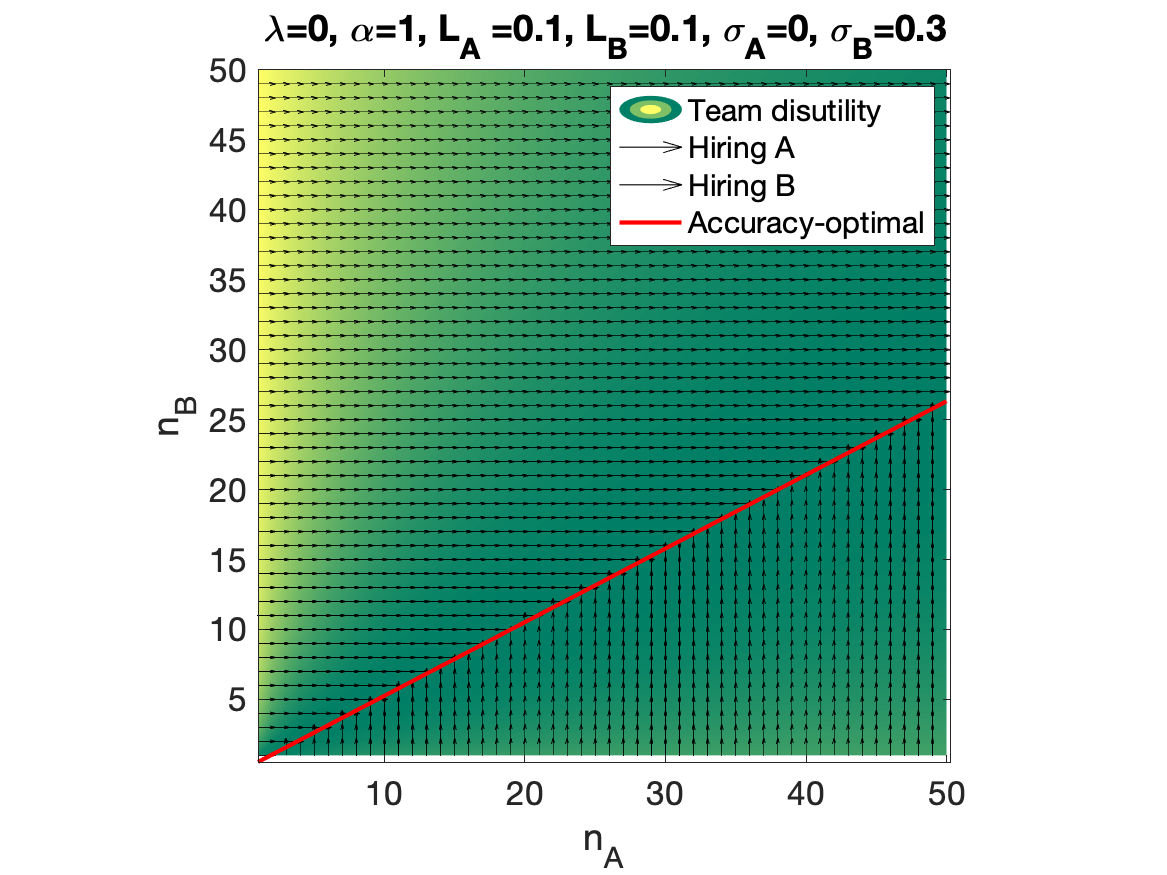}
    \end{subfigure}
    \caption{Visualization of team formation dynamics for $\lambda=0$ when $\alpha=1$ and $\sigma_A=0$ and $\sigma_B \in \{0.1, 0.2, 0.3\}$. Even though individual members of each type are noisy, it is never simultaneously beneficial to add a new member of type A and a new member of type $B$.
    }\label{fig:plain_lambda_0_alpha_1_noise}
\end{figure*}

\begin{figure*}
    \centering
    \begin{subfigure}{0.32\textwidth}
        \includegraphics[width=\textwidth]{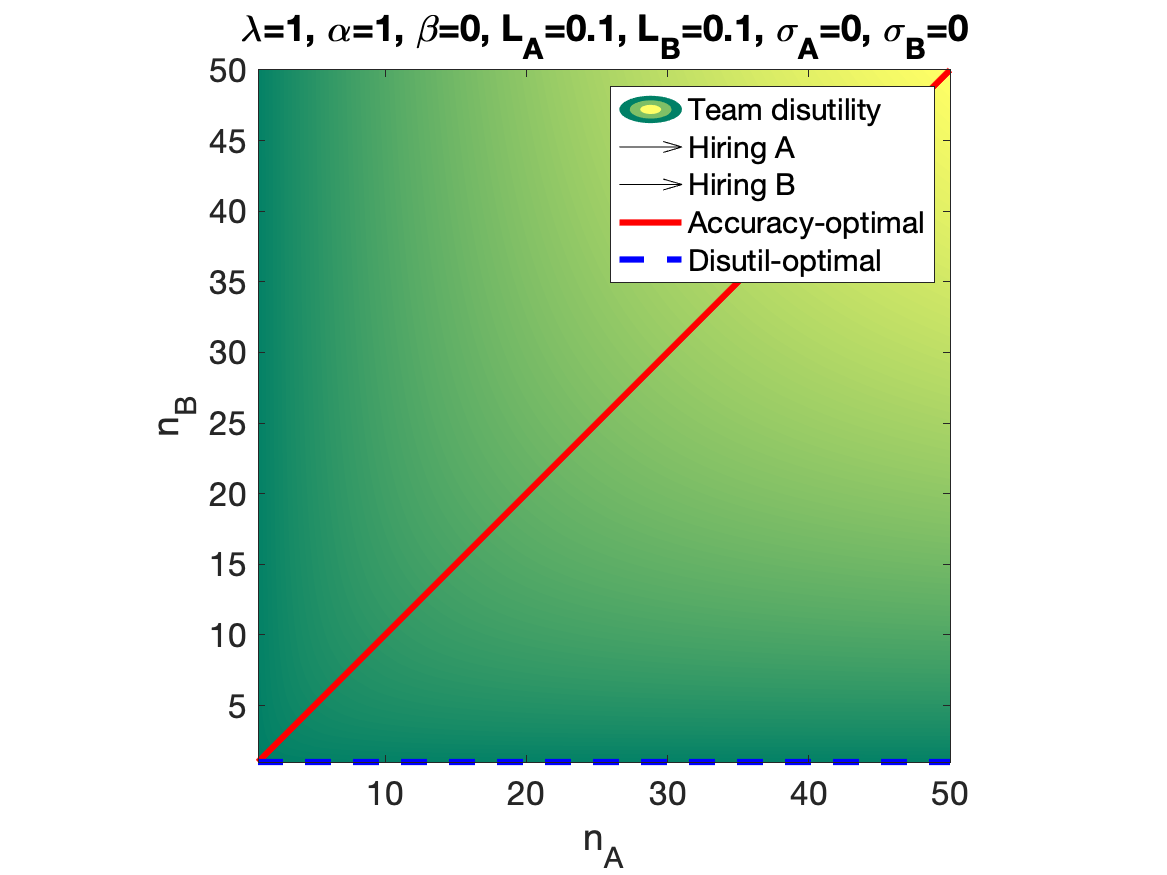}
    \end{subfigure}
    \begin{subfigure}{0.32\textwidth}
        \includegraphics[width=\textwidth]{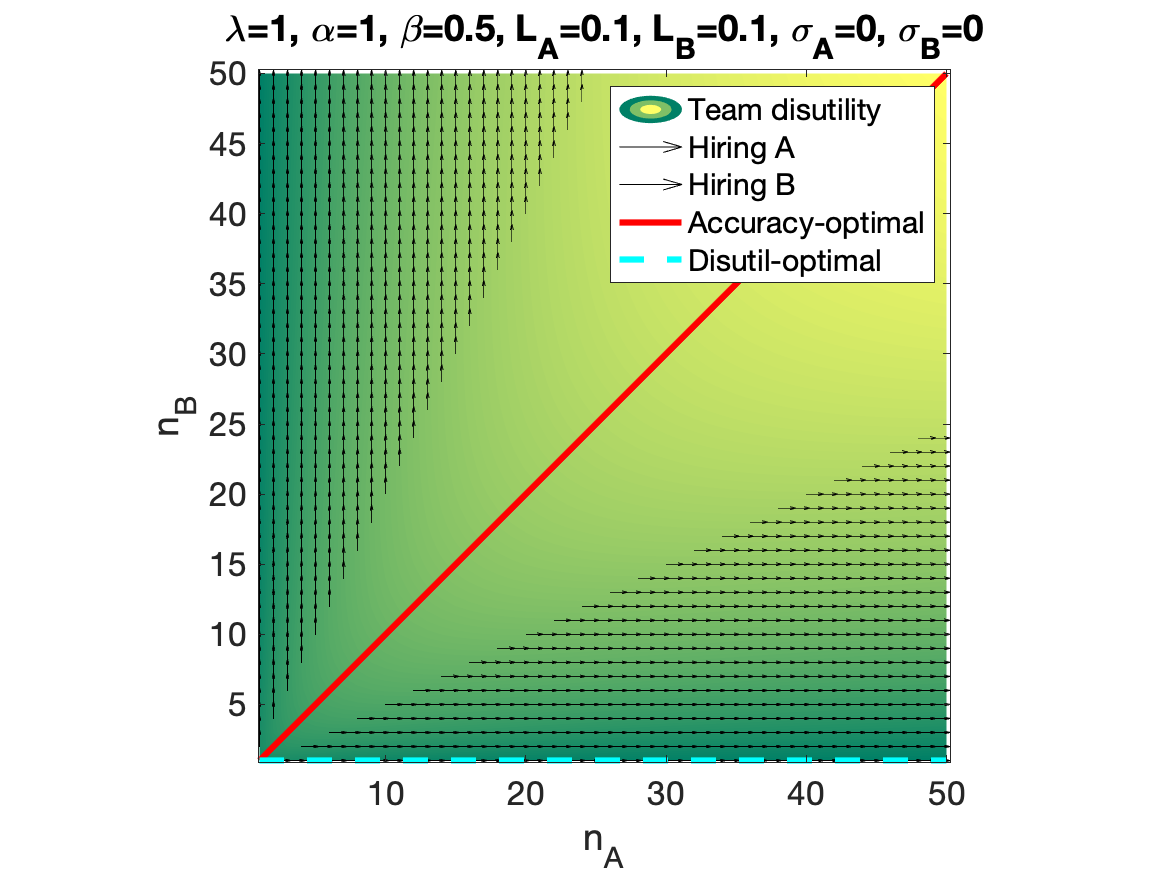}
    \end{subfigure}
    \begin{subfigure}{0.32\textwidth}
        \includegraphics[width=\textwidth]{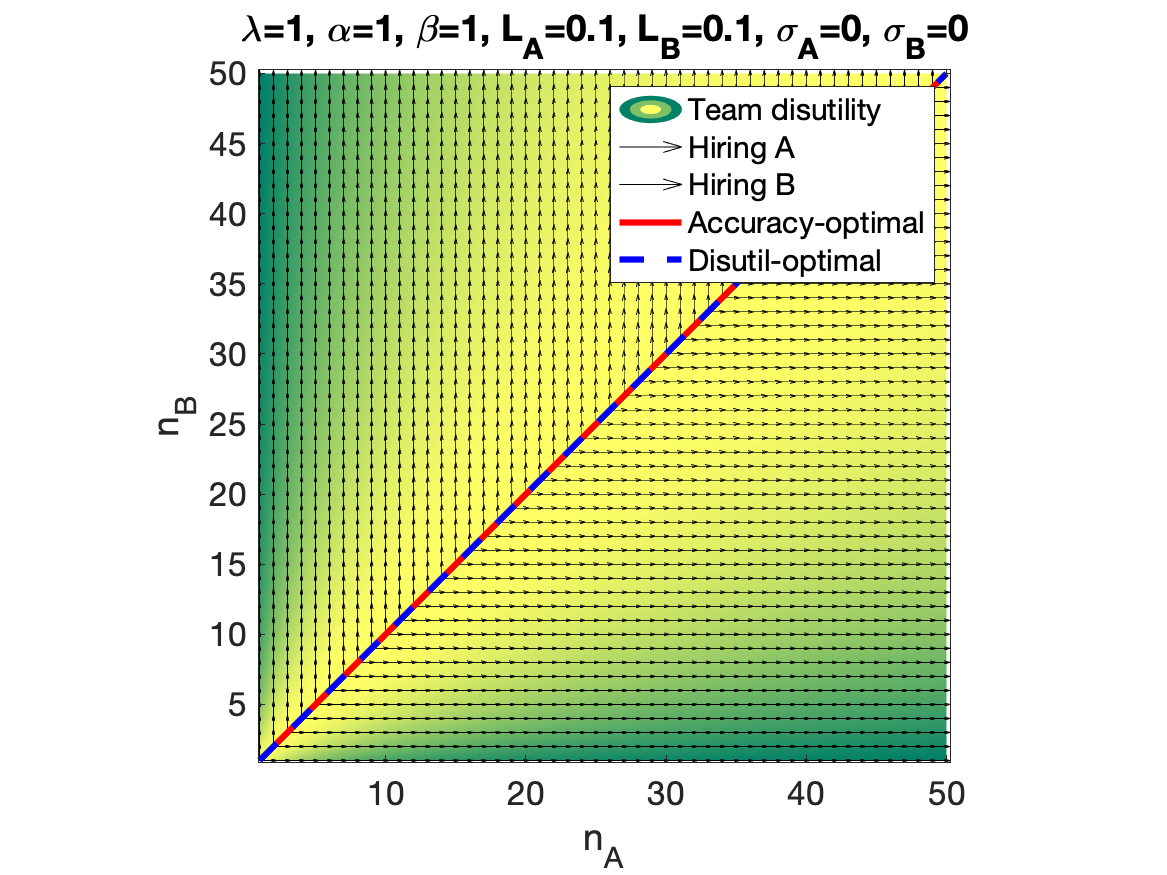}
    \end{subfigure}
    \caption{Visualization of team formation dynamics when $\lambda=1$ for several $\beta$ values (trends are similar for other values of $\cL^A, \cL^B$). Adding a new team member of the less-represented type is never beneficial. Adding a new member of the majority type only improves the team's disutility if $\beta$ is sufficiently large. 
    }
    \label{fig:plain_lambda_1}
\end{figure*}

\begin{figure*}[t!]
    \centering
    \begin{subfigure}[b]{0.32\textwidth}
        \includegraphics[width=\textwidth]{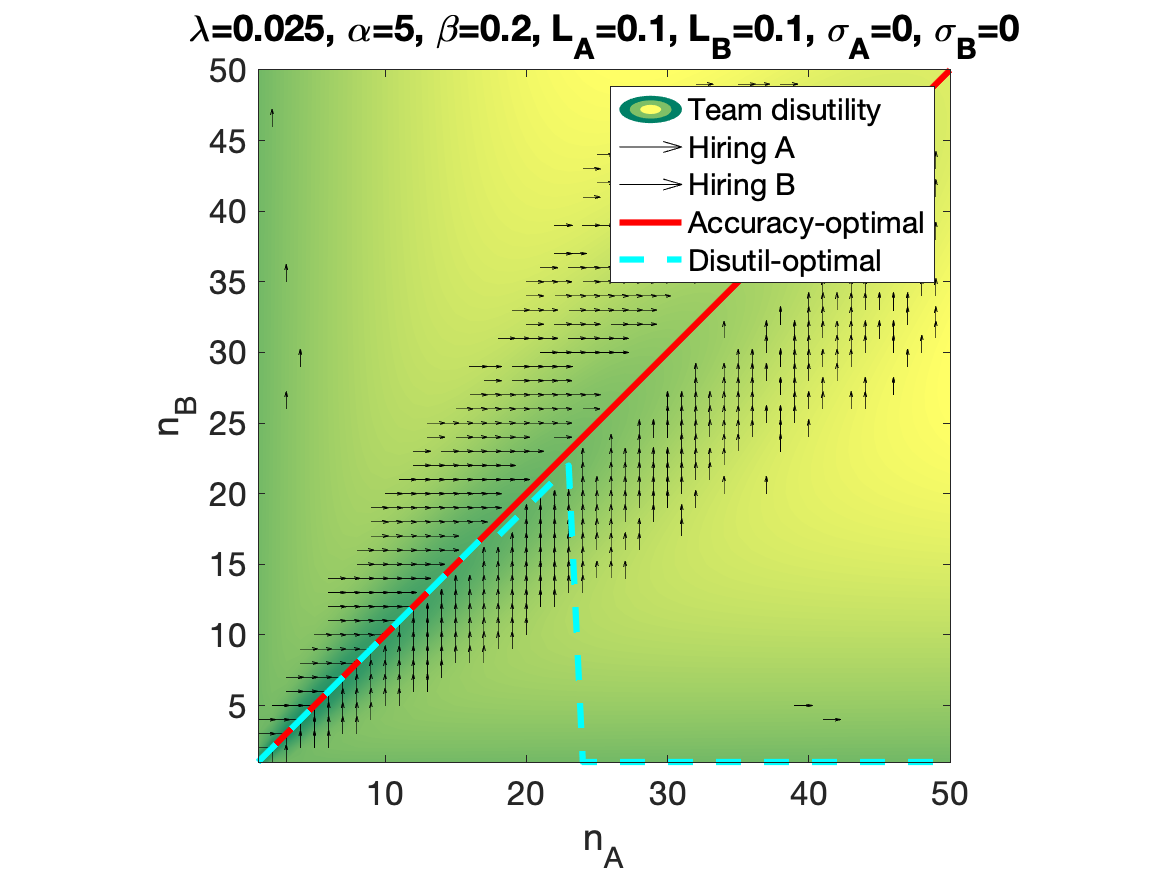}
    \end{subfigure}
    \begin{subfigure}[b]{0.32\textwidth}
        \includegraphics[width=\textwidth]{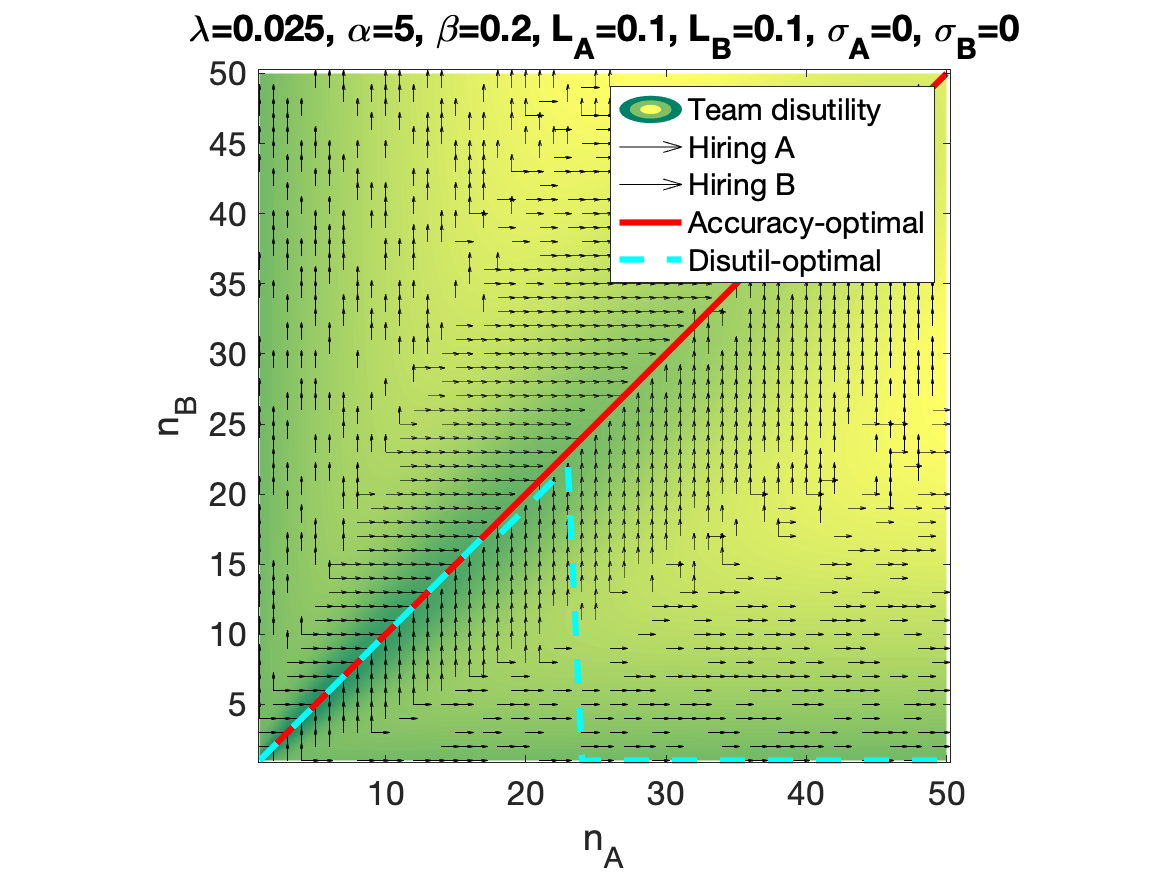}
    \end{subfigure}
    \begin{subfigure}[b]{0.32\textwidth}
        \includegraphics[width=\textwidth]{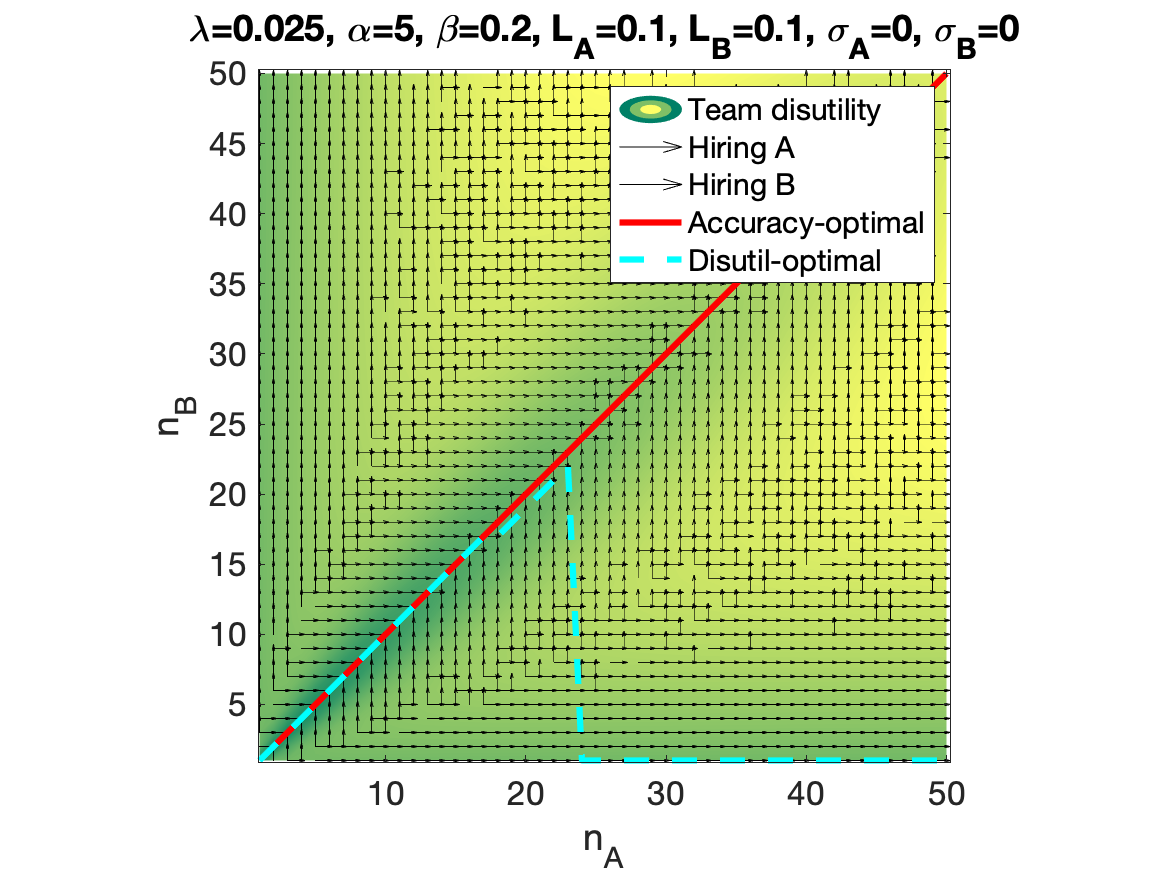}
    \end{subfigure}
    \caption{Visualization of team growth dynamics when assessments of accuracy gains are biased. When deciding whether to add a new member, the team's assessment of accuracy gains is corrupted by a bias term equal to  (a) $0.003$ (under-estimation of accuracy gains); (b) $0$ unbiased estimation of accuracy gains; (c) $-0.003$ (over-estimation of accuracy gains). When the estimation of accuracy gains are biased, the team may grow beyond accuracy optimal compositions, and adding a member of any type may be beneficial in certain compositions.}\label{fig:noisy_biased_assessment}
    \vspace{-3mm}
\end{figure*}

\hhedit{
\subsection{Extension to Three Types}\label{sec:three_types}

\xhdr{Error decomposition for three types} When $\lambda=0$, the team's mean-squared error can be decomposed into bias and variance terms as follows. All expectations are with respect to $(\vx,y) \sim P$ and $\epsilon_g \sim \cN(0,\sigma^2_g)$ for $g \in \{A,B,C\}$). For simplicity we will assume that all features are normalized such that $\Exp \left[ x_i^2\right] = 1$ for all $i$.

{\scriptsize
\begin{eqnarray*}
&&\Exp \left[ \left( \cG_T(\vx) - y \right)^2 \right]  =  \Exp \left[ \left( \cG_T(\vx) - f^*(\vx) \right)^2 \right] \\
&= & \Exp \left[ \left(\frac{n^\alpha_A}{n^\alpha_A + n^\alpha_B + n^\alpha_C} \theta^A(\vx) + \frac{n^\alpha_B}{n^\alpha_A + n^\alpha_B + n^\alpha_C} \theta^B(\vx) + \frac{n^\alpha_C}{n^\alpha_A + n^\alpha_B + n^\alpha_C} \theta^C(\vx) + \frac{n^\alpha_A \epsilon_A + n^\alpha_B\epsilon_B + n^\alpha_C\epsilon_C}{n^\alpha_A+n^\alpha_B +n^\alpha_C} - f^*(\vx) \right)^2 \right] \\
&= & \Exp \left[ \left(\frac{n^\alpha_A}{n^\alpha_A + n^\alpha_B + n^\alpha_C} \left(\theta^A(\vx) - f^*(\vx)\right) + \frac{n^\alpha_B}{n^\alpha_A + n^\alpha_B + n^\alpha_C} \left(\theta^B(\vx) - f^*(\vx)\right) + \frac{n^\alpha_C}{n^\alpha_A + n^\alpha_B + n^\alpha_C} \left(\theta^C(\vx) - f^*(\vx)\right) + \frac{n^\alpha_A \epsilon_A + n^\alpha_B\epsilon_B + n^\alpha_C\epsilon_C}{n^\alpha_A+n^\alpha_B +n^\alpha_C}  \right)^2 \right] \\
&=& \frac{n_A^{2\alpha}}{(n^\alpha_A + n^\alpha_B + n^\alpha_C)^2}  (\cL^A + \sigma^2_A)
+ \frac{n_B^{2\alpha}}{(n^\alpha_A + n^\alpha_B + n^\alpha_C)^2} (\cL^B + \sigma^2_B)
+ \frac{n_C^{2\alpha}}{(n^\alpha_A + n^\alpha_B + n^\alpha_C)^2} (\cL^C + \sigma^2_C)\\
&&+ \frac{n^\alpha_B n^\alpha_C}{(n^\alpha_A + n^\alpha_B + n^\alpha_C)^2} \Exp \left[\theta^A(\vx)^2\right]
+ \frac{n^\alpha_A n^\alpha_C}{(n^\alpha_A + n^\alpha_B + n^\alpha_C)^2} \Exp \left[\theta^B(\vx)^2\right]
+ \frac{n^\alpha_A n^\alpha_B}{(n^\alpha_A + n^\alpha_B + n^\alpha_C)^2} \Exp \left[\theta^C(\vx)^2\right],
\end{eqnarray*}
}
where in the last line we used the fact that $\text{cov}(x_i, x_j) = 0$ for all $j \neq i$, we know $\Exp \left[  \theta^A(\vx) \theta^B(\vx)   \right] = \Exp \left[  \theta^A(\vx) \theta^C(\vx)   \right] =\Exp \left[  \theta^C(\vx) \theta^B(\vx)   \right] =0$. Additionally, since noise terms are independent, we have $\Exp \left[\epsilon_A \epsilon_B \right] = \Exp \left[\epsilon_A \epsilon_C \right] =\Exp \left[\epsilon_C \epsilon_B \right] =0$.
Next, we translate the terms, $\Exp \left[\theta^g(\vx)^2\right]$, into a combination of loss terms, $\cL^g$'s, as follows.
\begin{eqnarray}\label{eq:three_types_A_loss}
    \cL^A &=& \Exp \left[ \left( f^*(\vx) - \theta^A(\vx)\right)^2 \right] \nonumber\\
    &=& \Exp \left[ \left( \theta^B(\vx) + \theta^C(\vx) \right)^2 \right] \nonumber\\
    &=& \Exp \left[ \theta^B(\vx) ^2 \right] + \Exp \left[ \theta^C(\vx)^2 \right]
\end{eqnarray}
where in the last line we used the fact that $\Exp \left[ x_i^2\right] = 1$ for all $i$.
With a similar logic, we obtain that:
\begin{equation}\label{eq:three_types_B_loss}
    \cL^B = \Exp \left[ \theta^A(\vx) ^2 \right] + \Exp \left[ \theta^C(\vx) ^2 \right]
\end{equation}
\begin{equation}\label{eq:three_types_C_loss}
    \cL^C = \Exp \left[ \theta^A(\vx) ^2 \right] + \Exp \left[ \theta^B(\vx) ^2 \right]
\end{equation}
Combing (\ref{eq:three_types_A_loss}), (\ref{eq:three_types_B_loss}), and (\ref{eq:three_types_C_loss}), we obtain:
\begin{equation}
    \Exp \left[ \theta^A(\vx) ^2 \right] = \cL^B + \cL^C - \cL^A
\end{equation}
\begin{equation}
    \Exp \left[ \theta^B(\vx) ^2 \right] = \cL^A + \cL^C - \cL^B
\end{equation}
\begin{equation}
    \Exp \left[ \theta^C(\vx) ^2 \right] = \cL^A + \cL^B - \cL^C
\end{equation}
Plugging in the above three equations into the error decomposition expression (), we obtain:
{\scriptsize
\begin{eqnarray*}
    \Exp \left[ \left( \cG_T(\vx) - y \right)^2 \right]  &=&  \frac{n_A^{2\alpha}}{(n^\alpha_A + n^\alpha_B + n^\alpha_C)^2}  (\cL^A + \sigma^2_A)
+ \frac{n_B^{2\alpha}}{(n^\alpha_A + n^\alpha_B + n^\alpha_C)^2} (\cL^B + \sigma^2_B)
+ \frac{n_C^{2\alpha}}{(n^\alpha_A + n^\alpha_B + n^\alpha_C)^2} (\cL^C + \sigma^2_C)\\
&&+ \frac{n^\alpha_B n^\alpha_C}{(n^\alpha_A + n^\alpha_B + n^\alpha_C)^2} \left(\cL^B + \cL^C - \cL^A\right)\\
&&+ \frac{n^\alpha_A n^\alpha_C}{(n^\alpha_A + n^\alpha_B + n^\alpha_C)^2} \left(\cL^A + \cL^C - \cL^B\right)\\
&&+ \frac{n^\alpha_A n^\alpha_B}{(n^\alpha_A + n^\alpha_B + n^\alpha_C)^2} \left(\cL^A + \cL^B - \cL^C\right),
\end{eqnarray*}
}

\xhdr{Derivation of disagreement term for three types}
When $\lambda=1$, the rate of disagreement can be computed as follows (all expectations are with respect to $(\vx,y) \sim P$ and $\epsilon_g, \epsilon'_g \sim \cN(0,\sigma^2_g)$ for $g \in \{A,B,C\}$):
\vspace{-2mm}

{\scriptsize
\begin{eqnarray*}
&&\frac{1}{(n_A + n_B + n_C)^{1+\beta}} \sum_{i,j \in T} d(i,j) = \frac{1}{(n_A + n_B + n_C)^{1+\beta}} \sum_{i,j \in T} \Exp \left[ (\theta_i(\vx) + \epsilon_i - \theta_j(\vx) - \epsilon_j)^2 \right]\\
&=&  \frac{2n_A n_B}{(n_A + n_B  + n_C)^{1+\beta}} \Exp\left[ (\theta^A(\vx) + \epsilon^A - \theta^B(\vx) - \epsilon^B)^2 \right] \\
&& + \frac{2n_A n_C}{(n_A + n_B  + n_C)^{1+\beta}} \Exp\left[ (\theta^A(\vx) + \epsilon^A - \theta^C(\vx) - \epsilon^C)^2 \right] \\
&& + \frac{2n_C n_B}{(n_A + n_B  + n_C)^{1+\beta}} \Exp\left[ (\theta^C(\vx) + \epsilon^C - \theta^B(\vx) - \epsilon^B)^2 \right] \\
&& + \frac{n_A (n_A-1)}{(n_A + n_B + n_C)^{1+\beta}} \Exp\left[ (\epsilon_A - \epsilon'_A)^2 \right] 
+ \frac{n_B (n_B-1)}{(n_A + n_B + n_C)^{1+\beta}} \Exp\left[ (\epsilon_B - \epsilon'_B)^2 \right] 
+ \frac{n_C (n_C-1)}{(n_A + n_B + n_C)^{1+\beta}} \Exp\left[ (\epsilon_C - \epsilon'_C)^2 \right]\\
&=&  \frac{2n_A n_B}{(n_A + n_B + n_C)^{1+\beta}} \Exp\left[(\theta^A(\vx)- \theta^B(\vx))^2 \right] 
+ \frac{2n_A n_B}{(n_A + n_B + n_C)^{1+\beta}} \Exp\left[ (\epsilon^A - \epsilon^B)^2 \right] \\
&&+\frac{2n_A n_C}{(n_A + n_B + n_C)^{1+\beta}} \Exp\left[(\theta^A(\vx)- \theta^C(\vx))^2 \right] 
+ \frac{2n_A n_C}{(n_A + n_B + n_C)^{1+\beta}} \Exp\left[ (\epsilon^A - \epsilon^C)^2 \right] \\
&&+ \frac{2n_C n_B}{(n_A + n_B + n_C)^{1+\beta}} \Exp\left[(\theta^C(\vx)- \theta^B(\vx))^2 \right] 
+ \frac{2n_C n_B}{(n_A + n_B + n_C)^{1+\beta}} \Exp\left[ (\epsilon^C - \epsilon^B)^2 \right] \\
&&+ \frac{n_A (n_A-1)}{(n_A + n_B + n_C)^{1+\beta}} \Exp\left[ (\epsilon_A - \epsilon'_A)^2 \right] 
+ \frac{n_B (n_B-1)}{(n_A + n_B + n_C)^{1+\beta}} \Exp\left[ (\epsilon_B - \epsilon'_B)^2 \right] 
+ \frac{n_C (n_C-1)}{(n_A + n_B + n_C)^{1+\beta}} \Exp\left[ (\epsilon_C - \epsilon'_C)^2 \right]\\
&=&  \frac{2n_A n_B}{(n_A + n_B + n_C)^{1+\beta}} \Exp\left[(\theta^A(\vx)- \theta^B(\vx) + f^*(\vx) - f^*(\vx))^2 \right]  \\
&& +\frac{2n_A n_C}{(n_A + n_B + n_C)^{1+\beta}} \Exp\left[(\theta^A(\vx)- \theta^C(\vx) + f^*(\vx) - f^*(\vx))^2 \right]  \\
&& +\frac{2n_C n_B}{(n_A + n_B + n_C)^{1+\beta}} \Exp\left[(\theta^C(\vx)- \theta^B(\vx) + f^*(\vx) - f^*(\vx))^2 \right]  \\
&& + \frac{ 2 n_A (n_B+n_C) + 2n_A (n_A-1)}{(n_A + n_B)^{1+\beta}} \Exp\left[ \epsilon_A^2 \right] + \frac{2 n_B (n_A + n_C) + 2n_B (n_B-1)}{(n_A + n_B)^{1+\beta}} \Exp\left[ \epsilon_B^2 \right]+ \frac{2 n_C (n_A + n_B) + 2n_C (n_C-1)}{(n_A + n_B)^{1+\beta}} \Exp\left[ \epsilon_B^2 \right]\\
&=&  \frac{2n_A n_B}{(n_A + n_B + n_C)^{1+\beta}} \left( \cL^B + \cL^A + \Exp \left[\theta^C(\vx)^2\right] \right)  \\
&&+ \frac{2n_A n_C}{(n_A + n_B + n_C)^{1+\beta}} \left( \cL^C + \cL^A + \Exp \left[\theta^B(\vx)^2\right] \right)\\
&&+ \frac{2n_C n_B}{(n_A + n_B + n_C)^{1+\beta}} \left( \cL^B + \cL^C + \Exp \left[\theta^A(\vx)^2\right] \right)
\\
&& + \text{noise terms as above}\\
&=&  \frac{2n_A n_B}{(n_A + n_B + n_C)^{1+\beta}} \left( 2\cL^A + 2\cL^B - \cL^C \right)  \\
&&+ \frac{2n_A n_C}{(n_A + n_B + n_C)^{1+\beta}} \left( 2\cL^A + 2\cL^C -\cL^B  \right)\\
&&+ \frac{2n_C n_B}{(n_A + n_B + n_C)^{1+\beta}} \left( 2\cL^A + 2\cL^B - \cL^C \right)
\\
&& + \text{noise terms as above}\\
\end{eqnarray*}
}

}

\end{document}